\documentclass[11pt, a4paper]{article}

\usepackage[utf8]{inputenc} 
\usepackage[T1]{fontenc}    
\usepackage[margin=1in]{geometry} 
\usepackage{amsmath, amssymb, amsthm} 
\usepackage{graphicx}       
\usepackage{hyperref}       
\usepackage{natbib}         
\usepackage{authblk}        
\usepackage{amsfonts}
\usepackage{enumitem}
\usepackage{float}

\hypersetup{
    colorlinks=true,
    linkcolor=blue,
    filecolor=magenta,      
    urlcolor=cyan,
    citecolor=blue,
}

\usepackage{metalogo} 
\usepackage{algpseudocodex} 
\usepackage{multirow}%
\usepackage{amsfonts}%
\usepackage{mathrsfs}%
\usepackage{textcomp}%
\usepackage{manyfoot}%
\usepackage{booktabs}%
\usepackage{listings}%
\usepackage{subcaption}
\usepackage[bbgreekl]{mathbbol}
\DeclareSymbolFontAlphabet{\mathbbl}{bbold}

\newtheorem{theorem}{Theorem}[section] 
\newtheorem{lemma}[theorem]{Lemma}      
\newtheorem{corollary}{Corollary}[theorem] 

\title{Extremile scalar-on-function regression}

\newcommand{\smallo}{{o}}
\newcommand{\bigo}{\mathcal{O}}
\newcommand{\ProjK}{\mathcal{P}_K}

\newcommand*\de{\mathop{}\!\mathrm{d}}
\DeclareMathOperator*{\argmin}{arg\,min}

\usepackage{mathtools}

\DeclarePairedDelimiter\floor{\lfloor}{\rfloor}

\author[1]{Maria Laura Battagliola}
\author[2]{Martin Bladt}
\affil[1]{Department of Statistics, Instituto Tecnol\'ogico Aut\'onomo de M\'exico, Mexico City, Mexico}
\affil[2]{Department of Mathematical Sciences, University of Copenhagen, Copenhagen, Denmark}
\date{\today}

\begin{document}

\maketitle

\begin{abstract}
Extremiles provide a generalization of quantiles which are not only robust, but also have an intrinsic link with extreme value theory. This paper introduces an extremile regression model tailored for functional covariate spaces. The estimation procedure turns out to be a weighted version of local linear scalar-on-function regression, where now a double kernel approach plays a crucial role. Asymptotic expressions for the bias and variance are established, applicable to both decreasing bandwidth sequences and automatically selected bandwidths. The methodology is then investigated in detail through a simulation study. Furthermore, we illustrate the method’s applicability with an analysis of the Berkeley Growth data, showcasing its performance in a real-world functional data setting.
\end{abstract}

\section{Introduction}
\label{sec:intro}

Extreme Value Theory (EVT) provides a robust framework for studying rare yet consequential events, making it indispensable for understanding phenomena characterized by extreme behaviors. By offering tools to detect and analyze such events, EVT is critical for applications ranging from risk management to mitigation strategies. For instance, extreme climate events have enormous impacts on the environmental and socioeconomic balance. Thus, EVT plays a pivotal role in addressing emerging global challenges.

Quantile regression \citep{koenker1978} is a foundational statistical method that estimates and predicts quantiles of responses conditional on covariates. Expectiles, introduced by \cite{newey1987}, represent another measure of extremes, though they lack the straightforward interpretability of quantiles. Efforts to link expectiles with generalized quantiles \citep{jones1994, bellini2014} have added to their theoretical development. Extremiles, introduced by \cite{daouia2019}, differ by representing the expected maximum or minimum of independent realizations from the conditional distribution of responses given covariates. Their natural interpretation aligns them conceptually with EVT, offering greater robustness than quantiles, particularly when data becomes scarce.

Another layer of complexity arises when studying continuous phenomena. Functional Data Analysis (FDA) has been a staple framework for analyzing such data, treating them as continuous functions \citep{zhang2011, bonner2014modeling, fraiman2014, suhalia2017, tapia2022, villani2024}. However, the infinite-dimensional nature of functional data introduces unique challenges, requiring targeted methodologies. Recent advancements have successfully integrated EVT and FDA, with scalar-on-function quantile regression \citep{guo2015, yu2016, li2022} and expectile regression \citep{mohammedi2021, Rachdi22} addressing these complexities.

In this paper, we propose a novel statistical model that extends extremile regression \citep{Daouia21} to accommodate functional covariates. Departing from classical functional regression approaches, we develop an interpretable and robust methodology based on a weighted version of scalar-on-function local linear regression \citep{baillo2009, barrientos2010, boj2010}. The weights depend on the extremile level and an estimate of the conditional CDF, which we obtain using Nadaraya--Watson-type estimators with bandwidth selection as per \cite{chagny2014}. 

Our methodology is applied to the classic Berkeley Growth dataset, extending the local linear functional regression analysis of \cite{ferraty22} on the same data. We move beyond their mean-based approach to explore the full conditional distribution, analyzing the relationship between childhood growth velocity profiles and the extremiles of final adult height. Our analysis shows how early growth dynamics can signal not just average outcomes, but also the potential for exceptionally tall or short stature. Moreover, we compare the results for boys and girls, uncovering very different growth behaviours between genders.

The manuscript is structured as follows: Section \ref{sec:preliminaries} provides an overview of extremile regression and scalar-on-function regression. Section \ref{sec:ExtrFunReg} presents the methodological framework, including asymptotic results (proofs in the appendix). Section \ref{sec:simulations} evaluates our approach via simulations, and Section \ref{sec:application} consists in an application of our method on the  Berkeley Growth data. Finally, Section \ref{sec:discussion} summarizes findings and outlines future research directions. The proofs of the theoretical results can be found in the appendix, and the code to run the analysis is provided in the supplementary materials.

\section{Preliminaries}
\label{sec:preliminaries}
This section provides some preliminaries on conditional extremiles and local linear regression with functional covariates. Its primary aim is to provide notation and intuition for the construction of our main model, which has elements of both approaches.
\label{sec:background}
\subsection{Extremile regression for multivariate covariates}
\label{sec:ext_reg_scalar}
Consider the estimation of conditional extremiles on non-functional covariates, cf. \cite{Daouia21}. Specifically, consider the random variables $Y\in \mathbb{R}$ and $X \in \mathbb{R}^d$, and call $$F(y|x)=\Pr(Y\le y|X=x)$$ the conditional CDF of $Y$ given $X=x$. Moreover, assume that $\mathbb{E}(|Y||X=x)<\infty$.  For any fixed $\tau \in (0,1)$, the \textit{conditional extremile of order $\tau$} of $Y$ given $X=x$ is defined as
\begin{equation}
    \label{eq:cond_extremile}
    \xi_\tau(x) = \argmin_{\theta \in \mathbb{R}} \mathbb{E} \big(J_\tau(F(Y|X))\cdot[|Y-\theta|^2 - |Y|^2] | X=x \big),
\end{equation}
where $J_\tau(\cdot) = K^{'}_\tau(\cdot)$, with
\begin{align*}
    K_\tau(t) = \begin{cases}
1-(1-t)^{s(\tau)} &\text{if $\tau\in(0,1/2]$}\\
t^{r(\tau)} &\text{if $\tau\in[1/2,1)$}
\end{cases}
\end{align*}
for $t\in[0,1]$, and $r(\tau) = s(1-\tau) = \log(1/2)/\log(\tau)$.
Moreover, we have that
\begin{equation}
    \label{eq:cond_extremile1}
    \xi_\tau(x) = \mathbb{E}\big[Y J_\tau(F(Y|X))| X=x\big].
\end{equation}

An alternative representation, first given by \cite{daouia2019} for the marginal extremile and extended to the conditional extremile \eqref{eq:cond_extremile1} by \cite{Daouia21}, is
\begin{align}
 \label{eq:cond_extremile1_maxdef}
   \xi_\tau(x) = \begin{cases}
\mathbb{E}\left(\min\left(Y^x_1,\dots,Y^x_{\floor*{s(\tau)}}\right)\right) &\text{if $\tau\in(0,1/2]$}\\
\mathbb{E}\left(\max\left(Y^x_1,\dots,Y^x_{\floor*{r(\tau)}}\right)\right) &\text{if $\tau\in[1/2,1)$}
\end{cases}
\end{align}
where $Y^x_1,\dots,Y^x_{\floor*{s(\tau)}}$ and $Y^x_1,\dots,Y^x_{\floor*{r(\tau)}}$ are $\floor*{s(\tau)}$ and $\floor*{r(\tau)}$ independent and identically distributed copies drawn from the conditional distribution of $Y$ given $X=x$. As $\tau \uparrow 1$, and thus $r(\tau) \to \infty$, the number of copies of $Y$ increases, and so does the average of their maximum. In a symmetrical way, when $\tau \downarrow 0$ and $s(\tau) \to \infty$, $\xi_\tau(x)$
decreases. 

Consider now an independent and identically distributed sample $(Y_i, X_i)_{i=1,..,n}$. A local linear approximation to \eqref{eq:cond_extremile}
is obtained for any $h>0$ by
\begin{equation}
\label{eq:est_pb}
 (\hat{\alpha}_\tau, \hat{\beta}_\tau) = \argmin_{\alpha,\beta \in \mathbb{R}} \sum_{i=1}^n J_\tau(\hat{F}(Y_i|x))\{Y_i - \alpha - \beta(x - X_i) \}^2 L \left(\frac{x - X_i}{h} \right),
\end{equation}
where $\hat{F}(\cdot|x)$ is any uniformly consistent estimator of $F(\cdot|x)$, and $\hat{\alpha}_\tau$ is the estimator of $\xi_\tau$.
Notice that \eqref{eq:est_pb} is a weighted least squares problem solving the locally linearized check function $\rho_\tau(v) = v(\tau - {1}(v \leq 0))$, using kernel $L$ with bandwidth sequence $h=h_n$ such that $h_n \to0$ as $n\to\infty$. The solutions to \eqref{eq:est_pb} turn out to have a closed form, and \cite{Daouia21} showed asymptotic unbiasedness and normality of the extremile estimator, under further regularity conditions.

\subsection{Local linear scalar-on-function regression}
\label{sec:llfr}
Consider now random variables $Y\in \mathbb{R}$ and a random element $X(\cdot) \in E$, where $(E, d)$ is a semi-metric space with semi-metric $d$. The choice of $d$ depends on the problem one aims at solving, as is explained in detail in \cite[Chapter 3]{Ferraty06}. However, in most applications $E=L^2(\mathcal{S})$, where $\mathcal{S} \subset \mathbb{R}$, and $d^2(X_1, X_2) ={ \int_S (X_1(s) - X_2(s))^2 \text{d}s}$ for any $X_1(\cdot),X_2(\cdot) \in E$; we also adopt such functional space. Hence, $\langle X_1, X_2 \rangle = \int_\mathcal{S} X_1(s)X_2(s) \text{d}s$ and $||X_1||^2 = {\int_\mathcal{S} X_1^2(s) \text{d}s}$. 

A local linear mean regression model for functional regressors and scalar responses was proposed in \cite{ferraty22}. Such type of regression is nonparametric and models a linear relationship between response and covariates in a subset of $E$. In particular, given $x(\cdot)\in E$ and independent and identically distributed observations $(Y_i, X_i(\cdot))_{i=1,..,n}$, the scalar-on-function local linear mean regression model
around $x$ is given by
\begin{equation}
\label{eq:llfr}
   (\hat{\alpha}, \hat{\beta}) = \argmin_{\alpha\in \mathbb{R}, \,\beta \in L^2(\mathcal{S})} \sum_{i=1}^n\{Y_i - \alpha - \langle \beta, x-X_i \rangle \}^2 L \left(\frac{||x - X_i||}{h} \right).
\end{equation}
The minimizers in \eqref{eq:llfr} correspond to the estimators of the local linear mean and its functional derivative in $x(\cdot)$, namely $\hat{\alpha}$ and $\hat{\beta}$, respectively. Notice that, under the assumption that the first order Taylor expansion of the estimator around $x(\cdot)$ is valid, the functional derivative is the local linear approximation of the mean regression operator evaluated in a neighborhood of $x$. Hence, analysing the properties of $\hat{\beta}$ leads to a better understanding of the robustness of the estimator $\hat{\alpha}$. Notice that even though $\hat{\alpha}, \hat{\beta}$ depend on $x(\cdot)$, $h$ and $n$, we avoid specifying this dependence in the notation for ease of exposition.

To reduce the infinite-dimensional problem in \eqref{eq:llfr}, one can approximate the unknown functional coefficient $\beta$ by projecting it onto a finite-dimensional subspace. 
More specifically, one may take
\begin{equation}
\label{eq:beta_approx}
    \beta \approx \ProjK(\beta) = \sum_{k=1}^K b_k \phi_k,
\end{equation}
where $\ProjK(\cdot)$ indicates the projection onto subspace $\mathcal{F}^K$ of $L^2(\mathcal{S})$, and $\{ \phi_k(\cdot)\}_{k=1}^K$ are $K$ orthonormal basis functions spanning $\mathcal{F}^K$, with $ b_k = \langle \phi_k, \beta \rangle $. 
While various orthonormal bases like B-splines or Fourier series can be used, \cite{ferraty22} highlight the advantages of a data-driven approach using Functional Principal Component Analysis. In this setting, the basis functions $\{ \phi_k(\cdot)\}_{k=1}^K$ are the eigenfunctions corresponding to the $K$ largest eigenvalues of the empirical covariance operator of the sample curves $\{X_i(\cdot) \}_{i=1}^n$. This choice is powerful because it provides an optimal basis for representing the primary modes of variation in the data.

We can now consider finite dimensional version of \eqref{eq:llfr} which only relies on scalar coefficients, namely
\begin{equation}
    \label{eq:llfr_approx}
 (\hat{\alpha}, \hat{b}_1,\dots,\hat{b}_K)   =  \argmin_{\alpha,\, b_1,\dots, b_K \in \mathbb{R}} \sum_{i=1}^n \Big\{Y_i - \alpha - \sum_{k=1}^K b_k \langle \phi_k, x-X_i \rangle \Big\}^2 L \left(\frac{||x - X_i||}{h} \right).
\end{equation}

\section{Extremile scalar-on-function regression}
\label{sec:ExtrFunReg}


In this section we introduce our model, which essentially combines the methodologies presented in Section~\ref{sec:background}, defining a framework for extremile scalar-on-function regression. In particular, we derive the closed form of the extremile estimator of order $\tau$. Following the approach of \cite{ferraty22}, which highlights the theoretical challenges of the functional setting, we then study its fundamental asymptotic properties by deriving the rates of convergence for its bias and variance. Throughout this section we use and adapt technical conditions introduced in \cite{Ferraty10} and \cite{ferraty22}, see conditions (A1)-(A12) in the Appendix. It should be noted that the supplementary material of \cite{ferraty22} provides a detailed theoretical discussion of these conditions, including situations where they are fulfilled (e.g., for non-degenerate Gaussian processes) and an analysis showing that key assumptions hold in general situations. Proofs of all theorems are provided in the appendix.

\subsection{Setting and problem formulation}

Given the random variable $Y\in \mathbb{R}$ and random element $X(\cdot) \in E$, we use the same notation as in \eqref{eq:cond_extremile1} to write the conditional extremile, namely
$\xi_\tau(x) = \mathbb{E}\big[Y J_\tau(F(Y|X))| X=x\big],$
where $F(\cdot|x)$ the conditional CDF of the scalar response $Y$ given the functional covariate $X=x$. Notice that the expression \eqref{eq:cond_extremile1_maxdef} carries over to this framework as well, with  $Y^x_1,\dots,Y^x_{\floor*{s(\tau)}}$ and $Y^x_1,\dots,Y^x_{\floor*{r(\tau)}}$ independent copies from the conditional distribution of $Y$ given the functional $X=x$.

Consider now independent and identically distributed sample $(Y_i, X_i(\cdot))_{i=1,..,n}$. 
In order to adapt \eqref{eq:est_pb} to functional covariates, we require an estimator of the conditional CDF that uniformly converges to $F(\cdot|x)$. For this purpose, we use the Nadaraya--Watson-type estimator suggested by \cite{Ferraty10}, who showed that 
\begin{equation}
    \label{eq:F_hat_fct}
    \hat{F}_{h_{F}}(y|x) = \frac{\sum_{i=1}^n L_F \left(\frac{||x-X_i||}{h_F}\right) {1}(Y_i \leq y) }{\sum_{i=1}^n L_F \left(\frac{||x-X_i||}{h_F}\right)}.
\end{equation}
is a uniformly consistent estimator under conditions (A6)--(A10) in 
the appendix.

Now that an appropriate estimator for the conditional CDF is established, we formulate the finite dimensional minimisation problem, i.e.

\begin{align}
   & \label{eq:est_pb_fct_approx}
  (\hat{\alpha}_\tau, \hat{b}_{\tau,1},\dots,\hat{b}_{\tau,K})\\
  & = \argmin_{\alpha, b_1,\dots, b_K \in \mathbb{R}} \sum_{i=1}^n J_\tau(\hat{F}_{h_{F}}(Y_i|x))\left\{Y_i - \alpha - \sum_{k=1}^K b_k \langle \phi_k, x-X_i \rangle \right\}^2 L \left(\frac{||x - X_i||}{h} \right),\nonumber
\end{align}
where we have used \eqref{eq:beta_approx} for the finite approximation of the functional coefficient $\beta(\cdot) \in L^2(\mathcal{S})$. 
Notice that \eqref{eq:est_pb_fct_approx} is a weighted version of \eqref{eq:llfr_approx} with weights that depend on extremile level $\tau$. In particular, $J_{0.5}(\hat{F}_{h_{F}}(Y_i|x))=1$ for every $i=1,\dots,n$, and so we may recover \eqref{eq:llfr_approx} as a special case corresponding to $\tau = 0.5$.

The explicit solutions can also be extended from the ones of extremile regression with scalar data. We define
$$\boldsymbol{\Phi} = \begin{pmatrix}
1 & \langle \phi_1, x-X_1 \rangle & \dots & \langle \phi_K, x-X_1 \rangle \\
    \vdots & \vdots & & \vdots\\
    1 &   \langle \phi_1, x-X_n \rangle & \dots &  \langle \phi_K, x-X_n \rangle
\end{pmatrix},$$
and hence the solution takes form
\begin{equation}
\label{eq:WLS}
    \begin{pmatrix} \hat{\alpha}_{\tau}\\  \hat{b}_{\tau,1} \\ \vdots \\ \hat{b}_{\tau,K} \end{pmatrix} = (\boldsymbol{\Phi}^{\top} \mathbf{W}_\tau \boldsymbol{\Phi})^{-1}\boldsymbol{\Phi}^{\top} \mathbf{W}_\tau\mathbf{Y},
\end{equation}
where $\mathbf{Y} = (Y_1,\dots,Y_n)^{\top}$ and $\mathbf{W}_\tau= \mathbf{W} \mathbf{J}_\tau(F(\mathbf{Y}|x))$, with $\mathbf{W} = \mbox{diag} \left( L(h^{-1}||x - X_i|| )\right)_{i=1,..,n}$ and  $\mathbf{J}_\tau(F(\mathbf{Y}|x)) = \mbox{diag}(J_\tau(F(Y_i|x)))_{i=1,..,n}$. 
For the subsequent asymptotic analysis in Sections~\ref{sec:bias} and \ref{sec:var}, we use the true conditional CDF, $F$ in the definition of the weight matrix $\mathbf{W}_\tau$. This is a standard approach in such two-stage estimation problems. Under our assumption that $\hat{F}_{h_{F}}$ is a uniformly consistent estimator (Assumptions A8-A12), the error introduced by this plug-in estimation of the conditional CDF is of a smaller asymptotic order and does not affect the leading-order terms of the bias and variance we derive.

From \eqref{eq:WLS} trivially follows that
\begin{equation}
   \label{eq:alpha_WLS} 
   \hat{\alpha}_{\tau} = \boldsymbol{e}^{\top}({\boldsymbol{\Phi}}^{\top} \mathbf{W}_\tau \boldsymbol{\Phi})^{-1}{\boldsymbol{\Phi}}^{\top} \mathbf{W}_\tau \mathbf{Y},
\end{equation}
and
\begin{equation}
 \label{eq:bs_WLS} 
    \begin{pmatrix} \hat{b}_{\tau,1} \\ \vdots \\ \hat{b}_{\tau,K} \end{pmatrix} = [\boldsymbol{0}_K | I_K]({\boldsymbol{\Phi}}^{\top} \mathbf{W}_\tau \boldsymbol{\Phi})^{-1}{\boldsymbol{\Phi}}^{\top} \mathbf{W}_\tau \mathbf{Y},
\end{equation}
where $\boldsymbol{0}_K = (0,\dots,0)^{\top}\in \mathbb{R}^K$, $\boldsymbol{e} = [1,\boldsymbol{0}_K^{\top}]^{\top} \in \mathbb{R}^{K+1}$ and $I_K$ is the identity matrix of dimension $K$. We denote with $[\boldsymbol{0}_K | I_K] \in \mathbb{R}^{K \times K+1}$ the matrix that has $\boldsymbol{0}_K$ as first column and coincides with $I_K$ for the remaining columns.


\subsection{Asymptotic bias}
\label{sec:bias}
Let the assumptions (A2)--(A5) stated in the appendix hold. From \eqref{eq:alpha_WLS}, we consider the conditional expectation of $\hat{\alpha}_\tau$ given the sample of covariates $\boldsymbol{X}$. Since the weight matrix $\mathbf{W}_\tau = \mathbf{W} \mathbf{J}_\tau$ is random in both $\boldsymbol{X}$ and $\mathbf{Y}$, the expectation conditional on $\boldsymbol{X}$, i.e.
$$\mathbb{E}(\hat{\alpha}_\tau|\boldsymbol{X}) = \mathbb{E} \left[ \boldsymbol{e}^{\top}({\boldsymbol{\Phi}}^{\top} \mathbf{W}_\tau \boldsymbol{\Phi})^{-1}{\boldsymbol{\Phi}}^{\top} \mathbf{W}_\tau \mathbf{Y} \Big| \boldsymbol{X} \right],$$
must be handled carefully.
The analysis of this expectation follows a two-step logic. First, we analyze the "inner" expectation of the product $\mathbf{W}_\tau \mathbf{Y}$ by averaging over the randomness of $\mathbf{Y}$ while holding $\boldsymbol{X}$ fixed. For each observation $i$, the $i$-th element of this vector is
$$\mathbb{E} \left[L(h^{-1}|x-X_i|) J_\tau(F(Y_i|x)) Y_i \Big| X_i \right] =  L(h^{-1}|x-X_i|)  \mathbb{E} \left[J_\tau(F(Y_i|x)) Y_i \Big| X_i \right] = L(h^{-1}|x-X_i|)  \xi_\tau(X_i),$$
 since the kernel term $L(\cdot)$ depends only on $X_i$, and thanks to the definition \eqref{eq:cond_extremile1}.
After this first step, the randomness from $\mathbf{Y}$ has been averaged out. The resulting expression we analyze for the asymptotic bias is therefore obtained by applying the "outer" matrix operator, which depends only on $\boldsymbol{X}$, to this new, simplified vector
$$\boldsymbol{e}^{\top}({\boldsymbol{\Phi}}^{\top} \mathbf{W}_\tau \boldsymbol{\Phi})^{-1}{\boldsymbol{\Phi}}^{\top} \mathbf{W} (\xi_\tau(X_1),\dots,\xi_\tau(X_n))^\top.$$
This handling of the conditional expectation, following the methodology of \cite{ferraty22}, allows us to proceed with the second-order Taylor expansion of $\xi_\tau(\text{X})$ around $x$.
Denote by $\mathcal{P}^\perp_K$ the projection onto the orthogonal component of $\mathcal{F}^K$, i.e. $\mathcal{P}^\perp_K f = \sum_{k>K} b_k \phi_k$ for any $f \in L^2(\mathcal{S})$, where we use the same notation as in \eqref{eq:beta_approx}, with $b_k = \langle f, \phi_k \rangle$. Relying on (A1), for any function $\text{X}(\cdot) \in E$ and fixed $\tau\in (0,1)$, the second order Taylor expansion of $\xi_\tau(\text{X})$ around $x$ is given a.s. by
\begin{align}
\label{eq:taylor}
    \xi_\tau(\text{X}) & = \xi_\tau(x) + \langle \xi^{'}_{\tau,x}  , \text{X}-x \rangle + \frac{1}{2}\langle  \xi^{''}_{\tau,\rho}(\text{X} -x),  \text{X} -x\rangle \\
    & = \xi_\tau(x) + \langle \ProjK\xi^{'}_{\tau,x}, \text{X}-x\rangle  + \langle \ProjK^\perp \xi^{'}_{\tau,x}, \text{X}-x\rangle  + \frac{1}{2}\langle  \xi^{''}_{\tau,\rho}(\text{X} -x),  \text{X} -x\rangle 
\end{align}
with $\rho = x + t(\text{X}-x)$, $t\in (0,1)$. Moreover, $\xi^{'}_{\tau,x}$ is the first functional derivative of the extremile regression operator with respect to $x$, $\xi^{''}_{\tau,\rho}$ is the second functional derivative.
Then, plugging in the sample of curves $X_1(\cdot),\dots,X_n(\cdot)$ in expansion \eqref{eq:taylor}, we get
\begin{equation}
\label{eq:eq:taylor_matx}
    \begin{pmatrix}  \xi_\tau(X_1)  \\ \vdots \\ \xi_\tau(X_n) \end{pmatrix} =  \boldsymbol{\Phi} \begin{pmatrix} \xi_\tau(x) \\ \nabla \xi_\tau(x) \end{pmatrix} + \begin{pmatrix}  \langle \mathcal{P}^\perp_K \xi^{'}_{\tau,x}, X_1 -x \rangle \\ \vdots \\ \langle \mathcal{P}^\perp_K \xi^{'}_{\tau,x}, X_n -x \rangle \end{pmatrix} + \frac{1}{2} R_{\tau,\rho, x} 
\end{equation}
where $ \nabla \xi_\tau(x) = (\langle \xi^{'}_{\tau,x}, \phi_1 \rangle, \dots, \langle \xi^{'}_{\tau,x}, \phi_K\rangle)^{\top}$, since $ \langle \ProjK \xi^{'}_{\tau,x}, X_i-x\rangle  = \sum_{k=1}^K \langle \xi^{'}_{\tau,x}, \phi_k\rangle \langle \phi_k ,X_i -x\rangle$. Moreover, $R_{\tau,\rho, x} = (R_{\tau,\rho, x, 1}, \dots, R_{\tau,\rho, x, n})^{\top} $, with $ R_{\tau,\rho, x, i} = \langle \xi^{''}_{\tau,\rho}(X_i -x), X_i -x \rangle$.   Hence, we obtain
\begin{equation}
\label{eq:bias_extr}
    \mathbb{E}(\hat{\alpha}_\tau|\boldsymbol{X})= \xi_\tau(x) + T_1 + \frac{1}{2}T_2
\end{equation}
with 
\begin{equation}
\label{eq:T1}
    T_1 =  \boldsymbol{e}^{\top}({\boldsymbol{\Phi}}^{\top} \mathbf{W}_\tau \boldsymbol{\Phi})^{-1}{\boldsymbol{\Phi}}^{\top} \mathbf{W} ( \langle \mathcal{P}^\perp_K \xi^{'}_{\tau,x}, X_1 -x \rangle \dots \langle \mathcal{P}^\perp_K \xi^{'}_{\tau,x}, X_n -x \rangle)^{\top}
\end{equation}
and
\begin{equation}
\label{eq:T2}
T_2 = \boldsymbol{e}^{\top}({\boldsymbol{\Phi}}^{\top} \mathbf{W}_\tau \boldsymbol{\Phi})^{-1}{\boldsymbol{\Phi}}^{\top} \mathbf{W} R_{\tau,\rho,x}.
\end{equation}
Note that the leading term $\xi_\tau(x)$ in \eqref{eq:bias_extr} is recovered because, while the matrix expression $(\boldsymbol{\Phi}^{T}\mathbf{W}_{\tau}\boldsymbol{\Phi})^{-1}\boldsymbol{\Phi}^{T}\mathbf{W}\boldsymbol{\Phi}$ does not simplify to the identity algebraically, its asymptotic behavior allows for this simplification. This is justified by the convergence results for the estimator's matrix components established in Lemma~\ref{lemma:elements_var} in the appendix. This technique, which shows that the random matrix components converge to their deterministic counterparts, is central to the proofs in \cite{ferraty22}.

\begin{theorem}[Asymptotic bias of $\hat{\alpha}_\tau$]
\label{th:asymt_bias}
Under conditions (A1) -- (A10)
\begin{equation}
     \mathbb{E}(\hat{\alpha}_\tau|\boldsymbol{X})= \xi_\tau(x) + \bigo_{\Pr}(h ||\ProjK^\perp \xi_\tau^{'}||) + \frac{1}{2}\bigo_{\Pr}(h^2).
\end{equation}
\end{theorem}

\subsection{Asymptotic variance}
\label{sec:var}
We now turn to the variance. A similar reasoning to the one used for the asymptotic bias in Section~\ref{sec:bias} applies here. The variance of the estimator, conditional on $\boldsymbol{X}$, is given by:
$$\textup{Var}(\hat{\alpha}_\tau|\boldsymbol{X}) = \textup{Var}\left( \boldsymbol{e}^{\top}({\boldsymbol{\Phi}}^{\top} \mathbf{W}_\tau \boldsymbol{\Phi})^{-1}{\boldsymbol{\Phi}}^{\top} \mathbf{W}_\tau \mathbf{Y} \Big| \boldsymbol{X} \right).$$
The leading asymptotic behavior of this expression is determined by the variance of the weighted response vector, $\mathbf{W}_\tau \mathbf{Y}$. Since the responses $Y_i$ are conditionally independent given $\boldsymbol{X}$, the covariance matrix $\textup{Var}(\mathbf{W}_\tau \mathbf{Y} | \boldsymbol{X})$ is diagonal. Its $i$-th diagonal element is given by $(L(h^{-1}\|x-X_i\|))^2 \textup{Var}(J_\tau(F(Y_i|x))Y_i | X_i)$, which simplifies to $(L(h^{-1}\|x-X_i\|))^2 \sigma^2_\tau(X_i)$ by definition. By using assumption (A5), we then arrive at the following expression:
\begin{align}
\label{eq:var1}
    \textup{Var}(\hat{\alpha}_\tau|\boldsymbol{X}) = [\sigma^2_\tau(x) + \smallo(1)]\boldsymbol{e}^{\top}({\boldsymbol{\Phi}}^{\top} \mathbf{W}_\tau \boldsymbol{\Phi})^{-1}{\boldsymbol{\Phi}}^{\top} \mathbf{W}^2  \boldsymbol{\Phi} ({\boldsymbol{\Phi}}^{\top} \mathbf{W}_\tau \boldsymbol{\Phi})^{-1}\boldsymbol{e},
\end{align}
where $\mathbf{W}^2 = \mbox{diag} \left( (L(h^{-1}||x - X_i|| ))^2\right)_{i=1,..,n}$.

\begin{theorem}[Asymptotic variance of $\hat{\alpha}_\tau$]
\label{th:asymt_variance}
Under conditions (A1) -- (A6)
$$ \textup{Var}(\hat{\alpha}_\tau|\boldsymbol{X}) = \bigo_{\Pr}( \lambda_K\{n \pi_x(h) \}^{-1}). $$
\end{theorem}
where $\lambda_K$ is the smallest eigenvalue of the $\Gamma$ matrix, defined in 
the appendix and  $\pi_x(h)=\Pr(||\text{X}-x||\le h)$.

\subsection{Extremile estimator based on local bandwidth}
\label{sec:local_band}

In the previous sections we have considered a global bandwidth $h$ that goes to zero as the sample size increases as dictated by (A3). We now consider two options for a data-driven local bandwidth selection for finite samples.

The first one we consider is inspired by the bandwidth selected via $k$-nearest neighbours (kNN) algorithm suggested by \cite{ferraty22} for mean scalar-on-function local linear regression. 
More specifically, consider a finite collection of bandwidths $H$, and call $B_h(x)$ the open ball of radius $h$ centered in $x$. For every $i=1,\dots,n$, the local bandwidth is
\begin{equation}
\label{eq:loc_band_kNN}
h_{i}^k = \inf \bigg \{h \in H: \sum_{j=1,\, j\neq i}^n \boldsymbol{1}_{X_j \in B_h(x)} = k \bigg \},
\end{equation}
namely the minimal radius of the ball centered in $x$ that contains exactly $k$ functional independent variables not including $X_i(\cdot)$. Then the extremile estimator based on the kNN local bandwidth is
\begin{equation}
\label{eq:alpha_loc_WLS} 
   \hat{\alpha}^\ell_{\tau} = \boldsymbol{e}^{\top}({\boldsymbol{\Phi}}^{\top} \mathbf{W}^\ell_{\tau} \boldsymbol{\Phi})^{-1}{\boldsymbol{\Phi}}^{\top} \mathbf{W}^\ell_{\tau} \mathbf{Y},
\end{equation}
where $\mathbf{Y} = (Y_1,\dots,Y_n)^{\top}$ and $\mathbf{W}^\ell_\tau= \mathbf{W}^\ell \mathbf{J}_\tau(F(\mathbf{Y}|x))$, and where we now have $\mathbf{W}^\ell = \text{diag}{ \left( L((h_{i}^k)^{-1}||x - X_i|| )\right)}_{i=1,..,n}$ and  $\mathbf{J}_\tau(F(\mathbf{Y}|x)) = \text{diag}{(J_\tau(F(Y_i|x)))}_{i=1,..,n}$. The following theorem shows the asymptotic behaviour of estimator \eqref{eq:alpha_loc_WLS}.
\begin{theorem}(Asymptotic behaviour of $\hat{\alpha}^\ell_\tau$)
\label{th:loc_est}
Under (A1)-(A3) and (A5)-(A7), we have
\begin{equation}
\label{eq:asym_loc_est}
     \hat{\alpha}^\ell_{\tau}(x) - \xi_{\tau}(x) = \bigo_{\Pr}\left(||\ProjK^\perp \xi'_{\tau,x} || \pi_x^{-1}\left( \frac{k}{n} \right) \right) + \bigo_{\Pr}\left(\left( \pi_x^{-1}\left( \frac{k}{n} \right) \right)^2 \right) + \bigo_{\Pr}(k^{-1/2}).
\end{equation}
\end{theorem}

The second option takes bandwidth \eqref{eq:loc_band_kNN} as a building block for an extremile level-dependent bandwidth. Using the same idea as in the supplementary materials of \cite{Daouia21}, for a fixed $\tau\in(0,1)$, the adaptive local bandwidth is given by
\begin{align}
    \label{eq:loc_band_tau}
    h_{\tau,i}^k &=  h_{i}^k\left[  V_{K_{\tau}\circ \Phi} \cdot(J_\tau(\tau) \cdot \phi(\Phi^{-1}(\tau)))^2\frac{4\tau(1-\tau)}{(\phi(\Phi^{-1}(\tau)))^2}\right]^{1/5} \\
    & = h_{i}^k\left[ 4\tau(1-\tau) V_{K_{\tau}\circ \Phi} (J_\tau(\tau))^2\right]^{1/5}\\
    & = h_{i}^k\left[ 4\tau(1-\tau) \left(\int_0^1 (\Phi^{-1}(t) - \mu_\tau)^2 J_\tau(t) \text{d}t \right) (J_\tau(\tau))^2\right]^{1/5}
\end{align}
where $\Phi^{-1}(\cdot)$ and $\mu_\tau = \int_0^1 \Phi^{-1}(t) J_\tau(t) \text{d}t$ are the quantile function and the true extremile of level $\tau$ of a standard Gaussian distribution, respectively. Notice that we used notation $V_{K_{\tau}\circ \Phi} = \left(\int_0^1 (\Phi^{-1}(t) - \mu_\tau)^2 J_\tau(t) \text{d}t \right)$.

\subsection{Bandwidth selection for conditional CDF estimator}
\label{sec:bndwdt_CDF}

In order to pick a data-driven bandwidth for estimator \eqref{eq:F_hat_fct} we rely on the work of \cite{chagny2014}, which is based on a trade-off between bias and variance. We briefly overview such selection method in this section.

Consider data $\{Y_i, X_i(\cdot) \}_{i=1}^n$ and bandwidth $h_F \in H_F$, where $H_F$ is a finite set of positive bandwidths with cardinality depending on the sample size $n$. Then, the small ball probability of function $\text{X}\in E$, drawn from the distribution of random variable $X(\cdot)$, of being in  a neighborhood centered in $x(\cdot) \in E $ of radius $h_F$, namely $$\pi_x(h_F) = \Pr(|| \text{X} - x || \leq h_F),$$ is approximated by $$\widehat{\pi}_x(h_F) = \frac{1}{n} \sum_{i=1}^n \boldsymbol{1}_{\{||X_i - x||\leq h_F\}},$$ assuming that the data in the sample are iid realisations from the same distribution of $X(\cdot)$. Moreover, let $\mathcal{I}_Y \subset \mathbb{R}$ the interval in which $Y_1, \dots, Y_n$ lie.
Then, bandwidth $h^{\text{opt}}_F$ for estimator $\hat{F}_{h_F}(\cdot|x)$ is chosen such that it minimizes a trade-off between an approximation of the bias and an upper bound of the variance of such estimator. Concretely, define
\begin{equation}
\label{eq:V_est}
    \widehat{V}(h_F, x) =\begin{cases}
   \kappa \frac{\ln n}{n \ln \widehat{\pi}_x(h_F)} & \text{if $\widehat{\pi}_x(h_F) \neq 0$},\\
+ \infty &\text{otherwise},
\end{cases}
\end{equation}
and 
\begin{equation}
\label{eq:A_est}
    \widehat{A}(h_F, x) = \max_{h_F^{'} \in H_F} \left[
    \left ( \int_{\mathcal{I}_Y} \left (\hat{F}_{h^{'}_{F}}(y|x) - \hat{F}_{h_F \vee h_F^{'}}(y|x) \right)^2\text{d}y\right)
    - \widehat{V}(h_F^{'}, x) \right]_+ .
\end{equation}
Then\begin{equation}
\label{eq:h_F_opt}
   h_F^{\text{opt}}(x) = \argmin_{h_F \in H_F} \left(\widehat{A}(h_F, x) +  \widehat{V}(h_F, x)\right),
\end{equation}
where $\kappa>0$. 
$\widehat{A}(h)$ and $\widehat{V}(h)$ are the approximations of bias term and upper bound of variance term, respectively. 

In order to compute the optimal bandwidth for the $i$-th covariate $X_i(\cdot)$, i.e. $h_F^{\text{opt}}(X_i)$ for $i=1,\dots,n$, we use a leave-one-out procedure, taking $Y_1,\dots,Y_{i-1}, Y_{i+1}, \dots, Y_n$ and $X_1(\cdot),\dots,X_{i-1}(\cdot), X_{i+1}(\cdot), \dots, X_n(\cdot)$ as sample, and centering \eqref{eq:V_est}, \eqref{eq:A_est} and \eqref{eq:h_F_opt} in $X_i(\cdot)$. Notice that, as explained in \cite[Section 6]{chagny2014}, the choice of $\kappa$ controls the bias-variance trade-off, and small values of the parameter lead to smaller values of $\{ h_F^{\text{opt}}(X_i)\}_i$.

\section{Simulations}
\label{sec:simulations}

In this section we test our proposed methodology with a simulation study. In particular, we generate response variables according to a scalar-on-function regression model as follows:
\begin{align}
\label{eq:model_sim}
Y_i = \beta_0 + \int_\mathcal{S} \beta(s) X_i(s) \text{d}s + \sigma(X_i)\epsilon_i,\hspace{5mm} i \in \{1,\dots,n \},  
\end{align}
with $\epsilon_i \stackrel{iid}{\sim} \mathcal{N}(0,\sigma^2_\epsilon)$ independent of $\{X_i(\cdot) \}_{i=1}^n$, and where $\sigma(X_i)$ controls the heteroskedasticity of the model. Across all scenarios we fix the error term's standard deviation $\sigma_\epsilon = 0.25$, 
the intercept $\beta_0 = 0$, and the functional coefficient $\beta(s) = 2 \cos(2 \pi s)$, with $s\in \mathcal{S} = [0,1]$.
Moreover, we pick the heteroskedastic relationship to be
\begin{equation}
\label{eq:het_rel}
    \sigma(X_i) = 1 +  \int_\mathcal{S} \lvert X_i(s) \rvert \text{d}s.
\end{equation}

In order to compute the true conditional extremile of level $\tau$ corresponding to the model, we may use the relationship between extremiles and quantiles. Given a level $\tau\in(0,1)$, the conditional extremile $\xi_\tau(x)$ can be computed from the conditional quantile $q_t(x)$ as $\xi_\tau(x) = \int_0^1 q_t(x) J_\tau(t) \text{d}t$, where $J_\tau(t)$ is the corresponding extremile weight function.
For model \eqref{eq:model_sim}, the conditional $t$-quantile is given by
\begin{align*}
     q_t(X_i) = \beta_0 + \int_\mathcal{S} \beta(s) X_i(s) \text{d}s + \sigma_\epsilon \sigma(X_i)\Phi^{-1}(t),
\end{align*}
where $\Phi(\cdot)$ is the CDF of a standard Gaussian random variable.
Consequently, the conditional extremile of level $\tau$ is
\begin{align}
\label{eq:sim_model}
\nonumber
    \xi_\tau(X_i) & = \int_0^1 q_t(X_i) J_\tau(t) \text{d}t\\
\nonumber
    & = \left(\beta_0 + \int_\mathcal{S} \beta(s) X_i(s) \text{d}s \right) \int_0^1  J_\tau(t) \text{d}t + \sigma_\epsilon \sigma(X_i)  \int_0^1 \Phi^{-1}(t) J_\tau(t) \text{d}t\\
    & = \beta_0 + \int_\mathcal{S} \beta(s) X_i(s) \text{d}s +\sigma_\epsilon \sigma(X_i) \xi_{\epsilon^*,\tau},
\end{align}
where $\xi_{\epsilon^*,\tau}$ is the extremile of order $\tau$ of a standard Gaussian distribution, and the last equality holds since $\int_0^1 J_\tau(t) \text{d}t = 1$. 

The functional covariates $X_i(s)$ are generated using two different schemes to assess the method's performance on signals with different characteristics. In \textit{Scenario A}, the covariates are constructed from a B-spline basis. Specifically, $X_i(s) = \sum_{k=1}^{5} c_{ik} f_k(s)$, where $\{f_k(s)\}_{k=1}^5$ are the five central basis functions from a 7-degree-of-freedom B-spline basis on $\mathcal{S}$.  The coefficients are drawn as $c_{i1}, c_{i2} \sim N(0, 0.5^2)$, $c_{i3} \sim N(0, 0.25^2)$, and $c_{i4}, c_{i5} \sim N(0, 0.05^2)$. In \textit{Scenario B}, covariates are generated from a basis of trigonometric and constant functions: $X_i(s) = \sum_{j=1}^{5} c_{ij} g_j(s)$. The basis functions are $g_1(s) = 1$, $g_2(s) = \sin(\pi s)$, $g_3(s) = \cos(10\pi s)$, $g_4(s) = \sin(30\pi s)$, and $g_5(s) = \cos(40\pi s)$.  The corresponding coefficients are drawn as $c_{i1} \sim N(0, 0.25^2)$, $c_{i2} \sim N(2, 1^2)$, $c_{i3} \sim N(0, 0.5^2)$, and $c_{i4}, c_{i5} \sim N(0, 0.05^2)$. For each simulated data, the basis functions used for the finite-dimensional projection are the eigenfunctions of the functional sample, obtained via FPCA. In particular, their number $K$ is determined such that they explain up to 95\% of the total variability of each sample.

We evaluate the performance of our proposed estimator over $B=300$ Monte Carlo repetitions. For each repetition $b \in \{1, \dots, B\}$ and for each extremile level $\tau$ in a predefined grid, we assess the estimation accuracy using the Mean Squared Error (MSE). The MSE for a single simulation run $b$ is calculated as:
\begin{align*}
\text{MSE}_\tau^{(b)} = \frac{1}{n} \sum_{i=1}^n \left(\xi_\tau(X_i) - \hat{\xi}_{\tau}^{(b)}(X_i)\right)^2,
\end{align*}
where $\xi_\tau(X_i)$ is the true conditional extremile for the $i$-th observation as defined in \eqref{eq:sim_model}, and $\hat{\xi}_{\tau}^{(b)}(X_i)$ is its corresponding estimate from our method in the $b$-th run. We then report the averaged MSE (AMSE) and its standard deviation (SD) across all $B$ repetitions, computed as:
\begin{align*}
\text{AMSE}_\tau = \frac{1}{B} \sum_{b=1}^B \text{MSE}_\tau^{(b)} \quad \text{and} \quad \text{SD(MSE}_\tau) = \sqrt{\frac{1}{B-1} \sum_{b=1}^B \left(\text{MSE}_\tau^{(b)} - \text{AMSE}_\tau\right)^2}.
\end{align*}
Furthermore, we investigate the non-crossing property of our estimator. We compute the empirical crossing rate, which is the proportion of the $B$ simulation runs in which the estimated conditional extremile curves cross for at least one pair of levels. A crossing event is recorded for repetition $b$ if, for any two levels $\tau_j < \tau_k$, there exists at least one observation $i$ such that $\hat{\xi}_{\tau_j}^{(b)}(X_i) > \hat{\xi}_{\tau_k}^{(b)}(X_i)$. We compare this crossing rate with that of estimated conditional quantiles at the same level. The procedure for estimating the latter is detailed in the Appendix. 

We establish a benchmark setting with a sample size of $n=200$ observations, where each function is discretized on a grid of $S=100$ points, and we set the parameter $\kappa=1$. To assess the sensitivity of our method, we then create additional settings by varying one parameter at a time from this benchmark configuration: we consider sample sizes of $n=100$ and $n=300$; grid sizes of $S=50$ and $S=200$; and parameter values of $\kappa=0.5$ and $\kappa=2$. For each of these settings, we compute the estimated conditional extremiles for a grid of levels $\tau \in \{0.1, 0.2, 0.3, 0.4, 0.5, 0.6, 0.7, 0.8, 0.9\}$. Unless specified otherwise, we use an Epanechnikov kernel for both the weighted local linear regression and the conditional CDF estimator.
For the bandwidth selection in our extremile estimator, we exclusively use the $\tau$-dependent adaptive bandwidth proposed in our methodology. This choice is deliberate, as it represents a fully automatized and theoretically-motivated procedure for local bandwidth selection. By adopting the complete method as intended by the original authors of extremile regression \citep{Daouia21}, our simulation study evaluates its practical performance, rather than testing it with other, more arbitrary bandwidth choices.

\begin{table}[h!]
\centering
\caption{Averaged Mean Squared Error (AMSE) for Scenario A. Standard deviations are in parentheses. For readability, the reported values have been multiplied by $10^3$.}
\label{tab:amse_scenario_a}
\begin{tabular}{@{}lccccccccc@{}}
\toprule
& \multicolumn{9}{c}{Extremile Level ($\tau$)} \\
\cmidrule(l){2-10}
Setting & 0.1 & 0.2 & 0.3 & 0.4 & 0.5 & 0.6 & 0.7 & 0.8 & 0.9 \\
\midrule
$n=200, S=100, \kappa=1$ & 7 (4) & 5 (3) & 3 (3) & 3 (2) & 3 (3) & 3 (2) & 3 (2) & 5 (3) & 7 (5) \\
\midrule
$n=100$ & 11 (7) & 7 (5) & 6 (4) & 5 (4) & 6 (4) & 5 (4) & 6 (5) & 8 (6) & 11 (8)\\
$n=300$ & 7 (4) & 5 (3) & 3 (2) & 3 (2) & 3 (2) & 3 (2) & 3 (2) & 5 (2) & 7 (4) \\
\midrule
$S=50$ & 7 (4) & 5 (3) & 3 (3) & 3 (2)& 3 (3) & 3 (2)& 3 (2) & 5 (3) & 7 (5) \\
$S=200$ & 8 (4) & 5 (3) & 3 (3) & 3 (2) & 3 (3) & 3 (2) & 3 (2) & 5 (3) & 7 (5) \\
\midrule
$\kappa=0.5$ & 7 (4) & 5 (3) & 3 (3) & 3 (2) & 3 (3) & 3 (2) & 3 (2)& 5 (3)& 7 (5)\\
$\kappa=2$ & 7 (4) & 5 (3) & 3 (3) & 3 (2) & 3 (3) & 3 (2) & 3 (2)& 5 (3)& 7 (5)\\
\bottomrule
\end{tabular}%

\end{table}

\begin{table}[h]
\centering
\caption{Averaged Mean Squared Error (AMSE) for Scenario B. Standard deviations are in parentheses. For readability, the reported values have been multiplied by $10^3$.}
\label{tab:amse_scenario_b}
\resizebox{\textwidth}{!}{
\begin{tabular}{@{}lccccccccc@{}}
\toprule
& \multicolumn{9}{c}{Extremile Level ($\tau$)} \\
\cmidrule(l){2-10}
Setting & 0.1 & 0.2 & 0.3 & 0.4 & 0.5 & 0.6 & 0.7 & 0.8 & 0.9 \\
\midrule
$n=200, S=100, \kappa=1$ & 32 (22) & 17 (15)& 16 (13)& 16 (14)& 18 (16)& 16 (14)& 16 (14)& 18 (18)& 30 (22)\\
\hline
$n=100$ & 66 (152) &34 (39) & 31 (32)& 32 (30)& 35 (31)& 31 (26)& 29 (25)& 30 (25)&50 (46) \\
$n=300$ & 31 (19) &18 (11) &16 (10) & 16 (11)& 18 (13)& 16 (11)& 16 (11)& 17 (11)& 27 (17)\\
\hline
$S=50$ & 32 (22) &17 (15) &16 (13) &16 (14) & 18 (16)& 16 (14)& 16 (14)& 18 (18)& 30 (21)\\
$S=200$  & 32 (22) &17 (15) &16 (13) &16 (14) & 18 (16)& 16 (14)& 16 (14)& 18 (16)& 30 (22)\\
\hline
$\kappa=0.5$ & 32 (22) &17 (15) &16 (13) &16 (14) & 18 (16)& 16 (14)& 16 (14)& 18 (18)& 30 (22)\\
$\kappa=2$  & 32 (22) &17 (15) &16 (13) &16 (14) & 18 (16)& 16 (14)& 16 (14)& 18 (18)& 30 (22)\\
\bottomrule
\end{tabular}%
}
\end{table}

The numerical results for the AMSE under Scenarios A and B are presented in Tables \ref{tab:amse_scenario_a} and \ref{tab:amse_scenario_b}, respectively. A primary finding is the substantial difference in performance between the two settings. The proposed method demonstrates significantly higher accuracy and precision in Scenario A, where the AMSE is roughly 4 to 6 times smaller than in Scenario B across all corresponding parameter configurations. This is further evidenced by the standard deviations, which are considerably larger in Scenario B, indicating that the estimates are not only less accurate but also less stable, particularly for the smaller sample size of $n=100$. Despite these differences in magnitude, several consistent patterns emerge within both scenarios. As expected, the estimator's performance improves with a larger sample size, as the AMSE and its standard deviation decrease when moving from $n=100$ to $n=300$. Conversely, the method shows a notable robustness to the number of grid points ($S$) and the choice of the tuning parameter ($\kappa$), with results remaining largely unchanged when these are varied. Finally, for both scenarios, the estimation error exhibits a symmetric pattern across the extremile levels, with the highest AMSE values occurring at the tails ($\tau=0.1$ and $\tau=0.9$), reflecting the inherent difficulty of estimation in regions with less data.

\begin{table}[h!]
\centering
\caption{Global crossing rates for the simulation study.}
\label{tab:crossing_rates}
\begin{tabular}{@{}l cc cc @{}}
\toprule
& \multicolumn{2}{c}{Scenario A} & \multicolumn{2}{c}{Scenario B} \\
\cmidrule(lr){2-3} \cmidrule(lr){4-5}
Setting & Extremile & Quantile & Extremile & Quantile \\
\midrule
$n=200, S=100, \kappa=1$ & 0.35 &0.98 & 1& 0.99\\
\hline
$n=100$ & 0.62& 1&1 & 0.99\\
$n=300$ & 0.38 & 0.88& 1& 0.99\\
\hline
$S=50$ & 0.35& 0.98& 1& 0.99\\
$S=200$ & 0.35 &  0.98& 1& 0.99\\
\hline
$\kappa=0.5$ & 0.35&0.98 &1 & 0.99\\
$\kappa=2$ &0.35 & 0.98&1 & 0.99\\
\bottomrule
\end{tabular}
\end{table}

The empirical crossing rates, presented in Table \ref{tab:crossing_rates}, reveal a key advantage of our proposed extremile estimator. In Scenario A, the extremile method demonstrates a markedly superior ability to maintain the non-crossing property, with a crossing rate of 0.35 in the benchmark case, compared to a near-certain crossing rate of 0.98 for the quantile-based approach. While both methods struggle in the more challenging Scenario B, with crossing rates at or near 1.0, it is crucial to note that our method performs no worse than the standard quantile estimator. Furthermore, the results show that the non-crossing property is primarily sensitive to the sample size, particularly for the extremile method in Scenario A, while remaining robust to changes in the grid size $S$ and the parameter $\kappa$. It is worth noting that this global crossing metric is particularly stringent, as neither estimation procedure is designed to enforce non-crossing simultaneously across a dense grid of curves. These findings underscore the practical benefit of our approach, as it significantly mitigates the common issue of curve crossing in functional data models without introducing any disadvantage in more complex settings.

\begin{table}[h!]
\centering
\caption{Comparison of Averaged Mean Squared Error (AMSE) for Gaussian and Epanechnikov kernels in the benchmark setting ($n=200, S=100, \kappa=1$) under Scenario A. For readability, reported values have been multiplied by $10^3$. Standard deviations are in parentheses.}
\label{tab:kernel_comparison}
\begin{tabular}{@{}lccccccccc@{}}
\toprule
& \multicolumn{9}{c}{Extremile Level ($\tau$)} \\
\cmidrule(l){2-10}
Kernel & 0.1 & 0.2 & 0.3 & 0.4 & 0.5 & 0.6 & 0.7 & 0.8 & 0.9 \\
\midrule
Epanechnikov & 7 (4) & 5 (3) & 3 (3) & 3 (2) & 3 (3) & 3 (2) & 3 (2) & 5 (3) & 7 (5) \\
Gaussian & 7 (4) & 4 (3) & 3 (2) & 2 (2) & 2 (2) & 2 (2) & 3 (2) & 4 (3) & 7 (4) \\
\bottomrule
\end{tabular}%

\end{table}

We now investigate the impact of the kernel function choice within the local linear regression step by repeating the benchmark simulation for Scenario A using a Gaussian kernel. The results, presented in Table \ref{tab:kernel_comparison}, show that the Gaussian kernel yields a slight improvement in performance. The AMSE is moderately lower for the central extremile levels, and the corresponding standard deviations are also slightly reduced, suggesting more stable estimates. Furthermore, the use of a Gaussian kernel also improves the non-crossing property, reducing the empirical crossing rate from 0.35 to 0.33.

\begin{table}[h!]
\centering
\caption{Averaged Prediction Mean Squared Error (APMSE) for the benchmark setting ($n=200, S=100, \kappa=1$) under Scenario A. For readability, reported values have been multiplied by $10^3$. Standard deviations are in parentheses.}
\label{tab:prediction_amse}
\begin{tabular}{@{}lccccccccc@{}}
\toprule
& \multicolumn{9}{c}{Extremile Level ($\tau$)} \\
\cmidrule(l){2-10}
Scenario & 0.1 & 0.2 & 0.3 & 0.4 & 0.5 & 0.6 & 0.7 & 0.8 & 0.9 \\
\midrule
$n=200, S=100, \kappa=1$ & 8 (5)& 5 (4)& 4 (3)& 3 (3)& 3 (4)& 3 (4)& 4 (4)& 5 (5)&8 (6)\\
\bottomrule
\end{tabular}%
\end{table}

In addition to assessing the in-sample fit, we conduct a final simulation study to evaluate the out-of-sample predictive performance of our estimator. For this analysis, we focus on the benchmark setting ($n=200, S=100, \kappa=1$) under Scenario A. In each of the $B=300$ repetitions, we randomly partition the data into a training set, containing 80\% of the observations, and a test set with the remaining 20\%. The model is fitted on the training data, and the resulting estimator is used to predict the conditional extremiles for the functional covariates in the test set. The predictive accuracy is quantified using the Prediction Mean Squared Error (PMSE). For a single simulation run $b$, the PMSE for a given level $\tau$ is calculated over the $n_{\text{test}}$ observations in the test set $\mathcal{I}_{\text{test}}^{(b)}$ as:
\begin{align*}
\text{PMSE}_\tau^{(b)} = \frac{1}{n_{\text{test}}} \sum_{i \in \mathcal{I}_{\text{test}}^{(b)}} \left(\xi_\tau(X_i) - \hat{\xi}_{\tau}^{(b)}(X_i)\right)^2,
\end{align*}
where $\hat{\xi}_{\tau}^{(b)}(X_i)$ is the predicted extremile for the new observation $X_i$. We then report the Averaged PMSE (APMSE) and its corresponding standard deviation across all $B$ repetitions, defined as:
\begin{align*}
\text{APMSE}_\tau = \frac{1}{B} \sum_{b=1}^B \text{PMSE}_\tau^{(b)} \quad \text{and} \quad \text{SD(PMSE}_\tau) = \sqrt{\frac{1}{B-1} \sum_{b=1}^B \left(\text{PMSE}_\tau^{(b)} - \text{APMSE}_\tau\right)^2}.
\end{align*}

The out-of-sample prediction results for the benchmark case are reported in Table \ref{tab:prediction_amse}. A comparison with the corresponding in-sample results from Table \ref{tab:amse_scenario_a} reveals a consistency in performance. The Averaged Prediction MSE (APMSE) values are only slightly higher than their in-sample counterparts, which is expected due to the model being evaluated on unseen data. The overall error pattern across the extremile levels and the magnitude of the standard deviations remain very similar. This  correspondence between the in-sample and out-of-sample performance provides compelling evidence of the model's ability to generalize well, suggesting that the conclusions drawn from the in-sample analysis are robust and reliably extend to predictive scenarios.

\section{Application to the Berkeley Growth data}
\label{sec:application}
In this section, we illustrate our proposed method using the Berkeley Growth data, extending the work of \cite{ferraty22}. We move beyond mean estimation to predict conditional extremiles of height at age 18 for both boys ($n_B=39$) and girls ($n_G=54$), conditional on their growth velocity profiles. 

The top row of Figure~\ref{fig:berk_curves} shows the individual and mean velocity curves, $\mu^B(t)$ and $\mu^G(t)$. While both groups exhibit a sharp velocity decrease in early years, a key difference emerges later: the boys' mean velocity continues a steady decline, whereas the girls' mean velocity stabilizes after age seven. A Functional Principal Component Analysis (FPCA) reveals that the primary mode of variation for each group is captured by their respective first eigenfunctions, $\varphi_1^B(t)$ and $\varphi_1^G(t)$, which explain roughly 65\% of variability for boys and 48\% for girls. The influence of these eigenfunctions is almost entirely concentrated in the first five years, representing the variability in how rapidly growth velocity decelerates during early childhood; their effect becomes negligible thereafter. The bottom row of Figure~\ref{fig:berk_curves} displays the growth velocity profiles constructed as $\mu^B(t) \pm 5\varphi_1^B(t)$ and $\mu^G(t) \pm 5\varphi_1^G(t)$, which represent these faster or slower growth decelerations and will now serve as the functional covariates for our analysis.

\begin{figure}[htbp]
    \centering
\includegraphics[width=\textwidth]{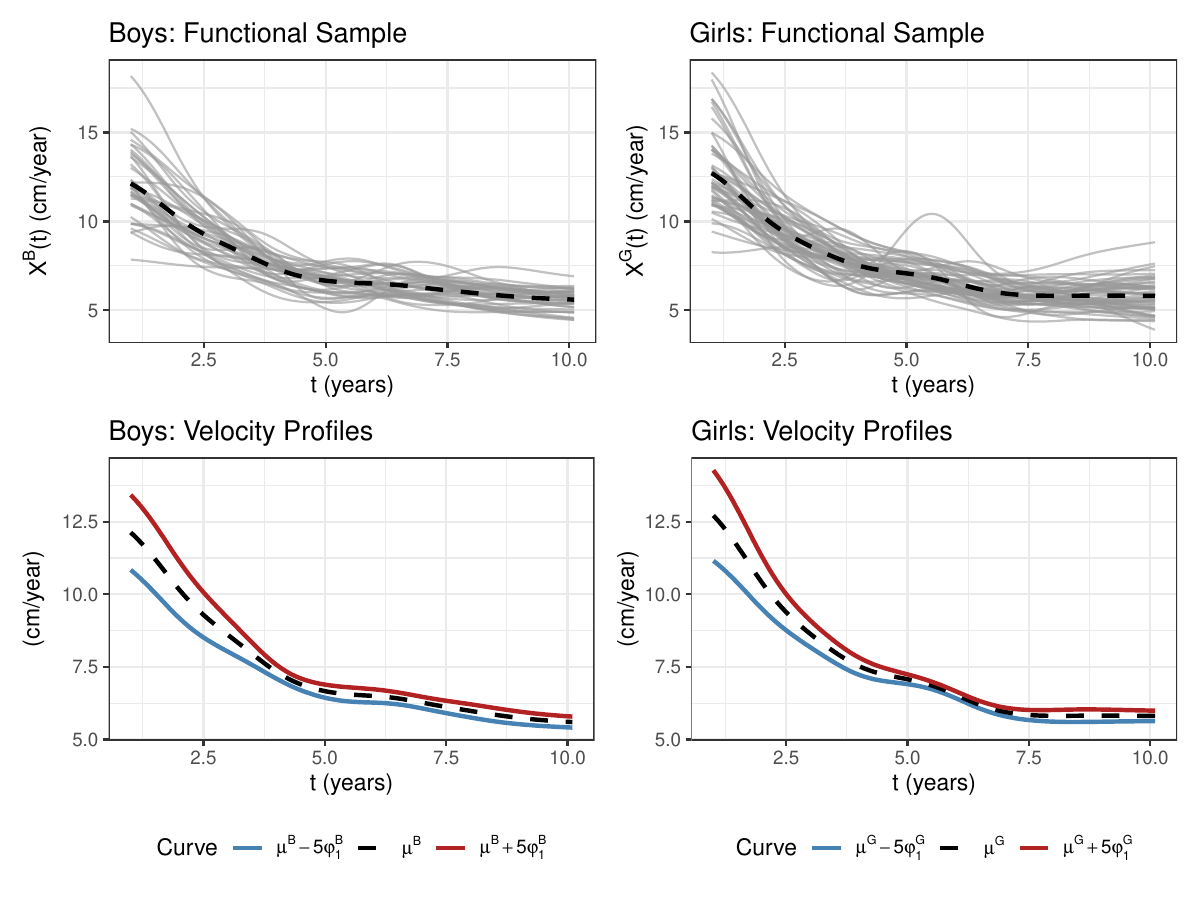}
    \caption{Top row: Growth velocity curves of $n_B=39$ boys (left) and of $n_G=54$ girls (right). The mean function of each sample is shown (dashed black line). Bottom row: Mean function (dashed black line), and variation to the mean given by the first eigenfunction (solid lines) for boys (left) and girls (right).}
    \label{fig:berk_curves}
\end{figure}

 Before analyzing the predicted conditional extremiles in our application, we recall their interpretation. While the extremile at level $\tau=0.5$ is equivalent to the conditional mean, extremiles for other levels provide information about the tails of the conditional distribution. Specifically, for a level $\tau>0.5$, the extremile represents the expected maximum value from a set of individuals sharing the same profile; for $\tau<0.5$, it represents the expected minimum. For instance, the extremile at $\tau=0.9$ corresponds to the expected maximum height among 6 individuals, at $\tau=0.1$ to the expected minimum among 6 individuals.

In Figure~\ref{fig:berk_extr} we show the predicted extremiles of height at 18 conditional on the velocity profiles of both groups. For boys, while all velocity profiles predict a similar mean height of around 180~cm, the pattern of predictions reverses from the lower to the upper tail. The profile with a slower initial deceleration, $\mu^B(t) - 5\varphi_1^B(t)$, is linked to the highest expected maximum height of approximately 186~cm at~$\tau=0.9$, but also the lowest expected minimum of 172~cm at~$\tau=0.1$. In contrast, the analysis for girls is straightforward and consistent across all levels: the profile with a faster initial deceleration, $\mu^G(t) + 5\varphi_1^G(t)$, is robustly associated with taller outcomes, predicting the highest minimums (around 162~cm), means (around 168~cm), and maximums (around 174~cm). 

Ultimately, this demonstrates the power of extremile regression: for boys, early growth velocity is not a simple predictor of height, but an indicator of potential and risk, while for girls, it is a direct marker of overall size. This nuanced relationship is completely invisible to a mean-only analysis.
\begin{figure}[htbp]
    \centering
\includegraphics[width=\textwidth]{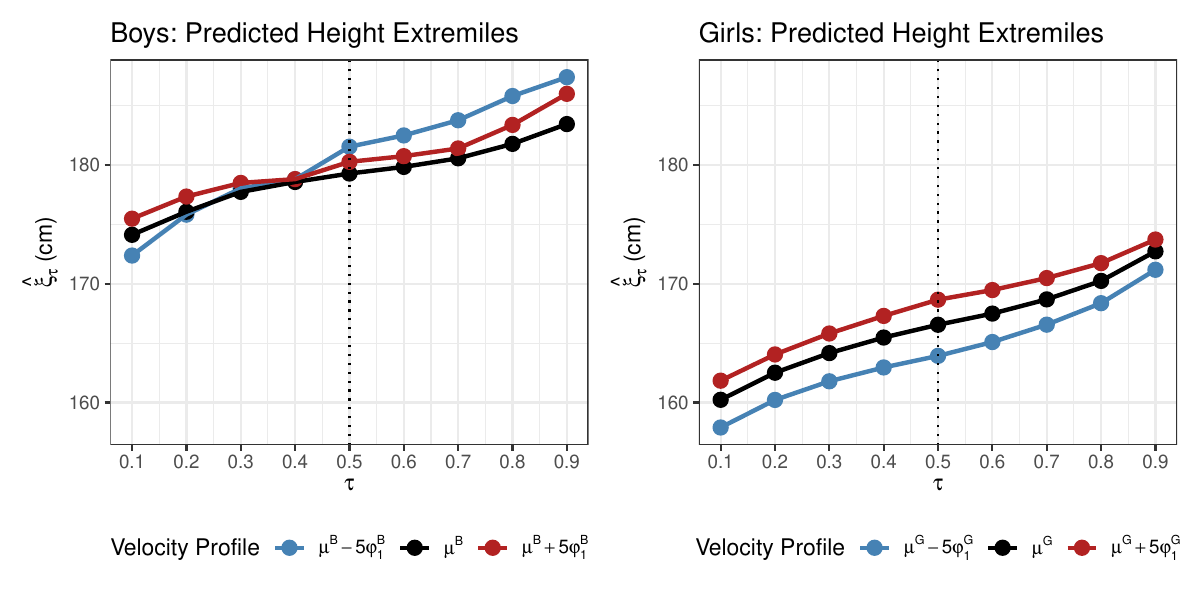}
    \caption{Predicted height extremiles corresponding to the three velocity profiles for boys (left) and girls (right). The vertical dotted line corresponds to the predicted local linear mean.}
    \label{fig:berk_extr}
\end{figure}

\section{Conclusion}
\label{sec:discussion}
In this work we adapted the scalar-on-function local-linear mean regression introduced by \cite{ferraty22} to extremile scalar-on-function regression. To do so, inspired by \cite{Daouia21}, we chose weights that are functions of the conditional CDF. We made the estimation as authomatized as possible by performing kernel estimation twice. First, we estimate the conditional CDF with the Nadaraya--Watson estimator adapted to functional covariates \citep{Ferraty10}. For the choice of local bandwith we use data-driven method introduced by \cite{chagny2014}. Second, we adapt the local bandwidth choice of \cite{ferraty22} by rescaling as \cite{Daouia21} suggests, so as to to make it extremile level-dependent. Notice that the choice of bandwidth in both cases is extremely important: in the former, it controls the extremile crossing, while in the latter it controls the ``influence" the sample has on the estimated extremile.

Thus, in addition to studying the asymptotic properties of the estimator, we tested the developed methodology with a simulation study that confirmed its practical effectiveness. A key finding relates to the common issue of curve crossing. Specifically, our estimator's performance is never worse than that of standard quantile techniques, and in some scenarios, it is substantially better. Moreover, the estimator proved robust to implementation choices, such as grid density and tuning parameters, while its strong out-of-sample performance underscored its ability to generalize well for predictive applications.
All the proposed methods are implemented in \texttt{R} \citep{Rmanual} with a package \texttt{ExtrFunReg}
 \citep{battagliola2024extrfunreg},
based on version 0.0.2 of package \texttt{fllr}\footnote{https://bitbucket.org/StanislavNagy/fllr/src/master/}, which is the implementation of the methodology introduced by \cite{ferraty22}. We include the scripts to generate all the plots in this work in the supplementary material.


Finally, our methodology was illustrated on the Berkeley Growth data, predicting adult height conditional on childhood growth velocity profiles. The analysis uncovered gender-specific relationships that a traditional mean regression would miss. For girls, a faster deceleration in early growth was directly associated with taller outcomes across all extremile levels. On the other hand, for boys a slower initial deceleration did not change the predicted mean height, but was linked to a wider range of potential outcomes, both the highest expected maximums and the lowest expected minimums. This highlights our method's ability to reveal how covariates can influence not just the center of a conditional distribution, but also its tails.

Concerning further developments, our work offers several avenues for extensions. Deriving the full asymptotic distribution of our estimator is a key next step. While this work establishes the estimator's consistency and convergence rates, a primary and challenging task in the functional setting, as noted by \cite{ferraty22}, the asymptotic distribution is necessary for constructing formal inference procedures like confidence intervals and hypothesis tests. This would likely involve extending the techniques from \cite{Daouia21} to the infinite-dimensional case, which represents a non-trivial but important avenue for future research.


Moreover, more complex functional models could be adapted to extremile regression. For instance, function-on-function models, where both responses and covariates are functions, and scalar-on-function mixed effects models for clustered or longitudinal data. Notice that such models are already implemented in the context of quantile regression (see for instance \cite{beyaztas2022} and \cite{battagliola2024}), and it would be worthwile to compare them to the corresponding extremile regression models, as was done in Section \ref{sec:simulations}.

Finally, it would be of great interest to extend the conditional extremile estimator to extreme levels $\tau \uparrow 1$ and $\tau \downarrow 0$, as done by \cite{Daouia21} in the case of multivariate covariates. Such extension to functional covariates, done by \cite{gardes2012} in the case of quantile scalar-on-function regression, needs to be handled with care, as operating with the distributions of functional random elements requires special mathematical care. In terms of threshold selection, methods from \cite{bladt2020threshold} could be useful in the functional setting. Nevertheless, given the amount of applications that can be represented via functional data, such methodology would be highly valuable. This could lead to more robust statistical models that are capable of dealing with extremes for a plethora of applications.

\section{Statements and Declarations}
The authors have no competing interests to declare that are relevant to the content of this article.

\section{Funding}
Martin Bladt was supported by the Carlsberg Foundation, grant number CF23-1096.

\bibliographystyle{plain}
\bibliography{main.bib}

@article{Rachdi22,
  author    = {Rachdi, Mustapha and Laksaci, Ali and Al-Kandari, Noriah M.},
  title     = {Expectile regression for spatial functional data analysis (sFDA)},
  journal   = {Metrika},
  year      = {2022},
  volume    = {85},
  number    = {5},
  pages     = {627--655},
  doi       = {10.1007/s00184-021-00846-x},
  url       = {https://doi.org/10.1007/s00184-021-00846-x}
}

@article{bladt2020threshold,
  title={Threshold selection and trimming in extremes},
  author={Bladt, Martin and Albrecher, Hansj{\"o}rg and Beirlant, Jan},
  journal={Extremes},
  volume={23},
  pages={629--665},
  year={2020},
doi={10.1007/s10687-020-00385-0},
  publisher={Springer}
}

@article{Daouia21,
  title={Extremile regression},
  author={Daouia, Abdelaati and Gijbels, Irene and Stupfler, Gilles},
  journal={Journal of the American Statistical Association},
  volume={117},
  number={539},
  pages={1579--1586},
  year={2022},
    doi={https://doi.org/10.1080/01621459.2021.1875837},
  publisher={Taylor \& Francis}
}

@book{Ferraty06,
author = {Ferraty, Fr\'{e}d\'{e}ric and Vieu, Philippe}, 
title = {Nonparametric Functional Data Analysis: Theory and Practice (Springer Series in Statistics)}, 
year = {2006}, 
isbn = {0387303693}, 
publisher = {Springer-Verlag}, 
doi={https://doi.org/10.1007/0-387-36620-2},
address = {Berlin, Heidelberg} }

@article{Ferraty10,
	author = {Fr{\'e}d{\'e}ric Ferraty and Ali Laksaci and Amel Tadj and Philippe Vieu},
	doi = {https://doi.org/10.1016/j.jspi.2009.07.019},
	issn = {0378-3758},
	journal = {Journal of Statistical Planning and Inference},
	number = {2},
	pages = {335-352},
	title = {Rate of uniform consistency for nonparametric estimates with functional variables},
	url = {https://www.sciencedirect.com/science/article/pii/S0378375809002316},
	volume = {140},
	year = {2010}}

@article{ferraty22,
    author = {Fr{\'e}d{\'e}ric Ferraty and Stanislav Nagy},
    title = "{Scalar-on-function local linear regression and beyond}",
    journal = {Biometrika},
    volume = {109},
    number = {2},
    pages = {439-455},
    year = {2021},
    month = {04},
    doi = {10.1093/biomet/asab027}}

@article{chagny2014,
	author = {Ga{\"e}lle Chagny and Angelina Roche},
	journal = {Electronic Journal of Statistics},
	number = {2},
	pages = {2352 -- 2404},
	title = {{Adaptive and minimax estimation of the cumulative distribution function given a functional covariate}},
	volume = {8},
    doi={https://doi.org/10.1214/14-EJS956},
	year = {2014}}

@article{daouia2019,
	author = {Daouia, Abdelaati and Gijbels, Irene and Stupfler, Gilles},
	journal = {Journal of the American Statistical Association},
	number = {527},
	pages = {1366--1381},
	title = {Extremiles: A New Perspective on Asymmetric Least Squares},
	volume = {114},
    doi={https://doi.org/10.1080/01621459.2018.1498348},
	year = {2019}}

@article{beyaztas2022,
	author = {Beyaztas, Ufuk and Shang, Han Lin and Alin, Aylin},
	journal = {Journal of Agricultural, Biological and Environmental Statistics},
	number = {1},
	pages = {149--174},
	title = {Function-on-Function Partial Quantile Regression},
	volume = {27},
    doi={https://doi.org/10.1007/s13253-021-00477-9},
	year = {2022}}

@article{battagliola2024,
	author = {Battagliola, Maria Laura and S{\o}rensen, Helle and Tolver, Anders and Staicu, Ana-Maria},
	date = {2024/02/06},
	doi = {10.1007/s13253-024-00601-5},
	isbn = {1537-2693},
	journal = {Journal of Agricultural, Biological and Environmental Statistics},
	title = {Quantile Regression for Longitudinal Functional Data with Application to Feed Intake of Lactating Sows},
	url = {https://doi.org/10.1007/s13253-024-00601-5},
	year = {2024}}

@article{koenker1978,
 ISSN = {00129682, 14680262},
 URL = {http://www.jstor.org/stable/1913643},
 author = {Roger Koenker and Gilbert Bassett},
 journal = {Econometrica},
 number = {1},
 pages = {33--50},
 publisher = {[Wiley, Econometric Society]},
 title = {Regression Quantiles},
 urldate = {2024-04-22},
 volume = {46},
 year = {1978},
doi={https://doi.org/10.2307/1913643}
}

@article{li2022,
	author = {Meng Li and Kehui Wang and Arnab Maity and Ana-Maria Staicu},
    doi={https://doi.org/10.1016/j.jmva.2022.104985},
	journal = {Journal of Multivariate Analysis},
	pages = {104985},
	title = {Inference in functional linear quantile regression},
	volume = {190},
	year = {2022}}

@article{guo2015,
	author = {Guo, Mengmeng and Zhou, Lan and Huang, Jianhua Z. and H{\"a}rdle, Wolfgang Karl},
	journal = {Statistics and Computing},
	number = {2},
	pages = {189--202},
	title = {Functional data analysis of generalized regression quantiles},
	volume = {25},
    doi={https://doi.org/10.1007/s11222-013-9425-1},
	year = {2015}}

@article{yu2016,
	author = {Dengdeng Yu and Linglong Kong and Ivan Mizera},
	journal = {Neurocomputing},
    doi={https://doi.org/10.1016/j.neucom.2015.08.116},
	pages = {74-87},
	title = {Partial functional linear quantile regression for neuroimaging data analysis},
	volume = {195},
	year = {2016}}

@article{barrientos2010,
	author = {Jorge Barrientos-Marin and Fr{\'e}d{\'e}ric Ferraty and Philippe Vieu},
	journal = {Journal of Nonparametric Statistics},
	number = {5},
	pages = {617--632},
	title = {Locally modelled regression and functional data},
	volume = {22},
    doi={https://doi.org/10.1080/10485250903089930},
	year = {2010}}

@article{boj2010,
	author = {Eva Boj and Pedro Delicado and Josep Fortiana},
	journal = {Computational Statistics \& Data Analysis},
	number = {2},
	pages = {429-437},
	title = {Distance-based local linear regression for functional predictors},
	volume = {54},
    doi={https://doi.org/10.1016/j.csda.2009.09.010},
	year = {2010}}

@article{baillo2009,
	author = {Amparo Ba{\'\i}llo and Aurea Gran{\'e}},
	journal = {Journal of Multivariate Analysis},
	number = {1},
	pages = {102-111},
	title = {Local linear regression for functional predictor and scalar response},
	volume = {100},
    doi={https://doi.org/10.1016/j.jmva.2008.03.008},
	year = {2009}}

@article{mohammedi2021,
	author = {Mustapha Mohammedi and Salim Bouzebda and Ali Laksaci},
	doi = {https://doi.org/10.1016/j.jmva.2020.104673},
	issn = {0047-259X},
	journal = {Journal of Multivariate Analysis},
	pages = {104673},
	title = {The consistency and asymptotic normality of the kernel type expectile regression estimator for functional data},
	url = {https://www.sciencedirect.com/science/article/pii/S0047259X20302542},
	volume = {181},
	year = {2021}}

@article{gardes2012,
	author = {Laurent Gardes and St{\'e}phane Girard},
	journal = {Electronic Journal of Statistics},
	number = {none},
	pages = {1715 -- 1744},
	title = {{Functional kernel estimators of large conditional quantiles}},
	volume = {6},
    doi={https://doi.org/10.1214/12-EJS727},
	year = {2012}}

@article{bonner2014modeling,
  title={Modeling regional impacts of climate teleconnections using functional data analysis},
  author={Bonner, Simon J and Newlands, Nathaniel K and Heckman, Nancy E},
  journal={Environmental and ecological statistics},
  volume={21},
  pages={1--26},
  year={2014},
  publisher={Springer},
doi={https://doi.org/10.1007/s10651-013-0241-8}
}

@article{fraiman2014,
	author = {Fraiman, Ricardo and Justel, Ana and Liu, Regina and Llop, Pamela},
	journal = {Canadian Journal of Statistics},
	number = {4},
	pages = {597-609},
	title = {Detecting trends in time series of functional data: A study of Antarctic climate change},
	volume = {42},
    doi={https://doi.org/10.1002/cjs.11231},
	year = {2014}}

@article{zhang2011,
	author = {Xianyang Zhang and Xiaofeng Shao and Katharine Hayhoe and Donald J. Wuebbles},
	doi = {10.1214/11-EJS655},
	journal = {Electronic Journal of Statistics},
	number = {none},
	pages = {1765 -- 1796},
	publisher = {Institute of Mathematical Statistics and Bernoulli Society},
	title = {{Testing the structural stability of temporally dependent functional observations and application to climate projections}},
	url = {https://doi.org/10.1214/11-EJS655},
	volume = {5},
	year = {2011}}

@article{tapia2022,
	author = {Mariela Tapia and Detlev Heinemann and Daniela Ballari and Edwin Zondervan},
    doi={https://doi.org/10.1016/j.renene.2022.03.049},
	journal = {Renewable Energy},
	pages = {1176-1193},
	title = {Spatio-temporal characterization of long-term solar resource using spatial functional data analysis: Understanding the variability and complementarity of global horizontal irradiance in Ecuador},
	volume = {189},
	year = {2022}}

@article{suhalia2017,
    doi={https://doi.org/10.1007/s00704-016-1778-x},
	author = {Suhaila, Jamaludin and Yusop, Zulkifli},
	journal = {Theoretical and Applied Climatology},
	number = {1},
	pages = {229--242},
	title = {Spatial and temporal variabilities of rainfall data using functional data analysis},
	volume = {129},
	year = {2017}}

@article{villani2024,
	author = {Villani, Veronica and Romano, Elvira and Mateu, Jorge},
    doi={https://doi.org/10.1007/s10651-024-00616-8},
	journal = {Environmental and Ecological Statistics},
	title = {Climate model selection via conformal clustering of spatial functional data},
	year = {2024}}

@article{bellini2014,
	author = {Fabio Bellini and Bernhard Klar and Alfred M{\"u}ller and Emanuela {Rosazza Gianin}},
	doi = {https://doi.org/10.1016/j.insmatheco.2013.10.015},
	issn = {0167-6687},
	journal = {Insurance: Mathematics and Economics},
	pages = {41-48},
	title = {Generalized quantiles as risk measures},
	url = {https://www.sciencedirect.com/science/article/pii/S0167668713001698},
	volume = {54},
	year = {2014}}

@article{newey1987,
 ISSN = {00129682, 14680262},
 URL = {http://www.jstor.org/stable/1911031},
doi={https://doi.org/10.2307/1911031},
 author = {Whitney K. Newey and James L. Powell},
 journal = {Econometrica},
 number = {4},
 pages = {819--847},
 publisher = {[Wiley, Econometric Society]},
 title = {Asymmetric Least Squares Estimation and Testing},
 urldate = {2024-05-16},
 volume = {55},
 year = {1987}
}

@article{jones1994,
	author = {Jones, M. Chris },
	date = {1994/05/27/},
	doi = {https://doi.org/10.1016/0167-7152(94)90031-0},
	isbn = {0167-7152},
	journal = {Statistics \& Probability Letters},
	number = {2},
	pages = {149--153},
	title = {Expectiles and M-quantiles are quantiles},
	url = {https://www.sciencedirect.com/science/article/pii/0167715294900310},
	volume = {20},
	year = {1994}}

@Manual{Rmanual,
    title = {R: A Language and Environment for Statistical Computing},
    author = {{R Core Team}},
    organization = {R Foundation for Statistical Computing},
    address = {Vienna, Austria},
    year = {2023},
    url = {https://www.R-project.org/},
  }

@misc{ExtrFunReg,
	author = {Maria Laura Battagliola},
	doi = {10.5281/zenodo.11400425},
	month = may,
	publisher = {Zenodo},
	title = {LauraBattagliola/ExtrFunReg: ExtrFunReg package},
	url = {https://doi.org/10.5281/zenodo.11400425},
	version = {v0.0.2},
	year = 2024}

@misc{battagliola2024extrfunreg,
  author       = {Battagliola, Maria Laura},
  title        = {{LauraBattagliola/ExtrFunReg}: {ExtrFunReg} package (v0.0.2)},
  year         = 2024,
  publisher    = {Zenodo},
  version      = {v0.0.2},
  doi          = {10.5281/zenodo.11400425},
  url          = {https://doi.org/10.5281/zenodo.11400425}
}

@Manual{refund,
    title = {refund: Regression with Functional Data},
    author = {Jeff Goldsmith and Fabian Scheipl and Lei Huang and Julia Wrobel and Chongzhi Di and Jonathan Gellar and Jaroslaw Harezlak and Mathew W. McLean and Bruce Swihart and Luo Xiao and Ciprian Crainiceanu and Philip T. Reiss and Erjia Cui},
    year = {2024},
    note = {R package version 0.1-37},
    url = {https://CRAN.R-project.org/package=refund},
  }

@Manual{quantreg,
    title = {quantreg: Quantile Regression},
    author = {Roger Koenker},
    year = {2024},
    note = {R package version 5.98},
    url = {https://CRAN.R-project.org/package=quantreg},
  }


\appendix                       
\section{Proofs}\label{app:A}

\subsection{Notation for proofs}
\label{sec:notation}

We keep part of the notation similar to the one of \cite{ferraty22} in order to make comparisons with their work straightforward.

In the following, we list the terms that appear in Lemmas \ref{lemma:S8} and \ref{lemma:S9}, and Corollary \ref{cor:upper_bound_alpha}:
\begin{align*}
\text{X}(\cdot) &\in E \text{ is a function drawn from the distribution of $X(\cdot)$}\\
    \pi_x(h)&=\Pr(||\text{X}-x||\le h), \text{ is continuous in a neighborhood of 0. }\\
    \vartheta_x(u)&=\lim_{h\downarrow 0}\frac{\pi_x(uh)}{\pi_x(h)}, \text{ for all } u>0\\
    \gamma^{p_1,\dots,p_K}_{j_1,\dots,j_K}(t)&=\mathbb{E}\left(\prod_{k=1}^K\langle \phi_{j_k},\text{X}-x \rangle^{p_k}\Big| ||\text{X}-x||^{\sum_k p_k}=t\right)\\
\tilde\gamma^{p_1,\dots,p_K}_{j_1,\dots,j_K}(t)&=\frac{\de}{\text{d}t}\gamma^{p_1,\dots,p_K}_{j_1,\dots,j_K}(t)\\
[\Gamma]_{j_1, j_2} &= \tilde\gamma^{1,1}_{j_1,j_2}(0)\\
\sigma^2(x) &= \text{Var}(Y | X=x)\\
\bbalpha^1_\tau(0|x,K) &= \sum_{k > K} \langle \phi_k, \xi^{'}_{\tau,x} \rangle \tilde{\gamma}_{k}^{1}(0)\\
\bbalpha^1_\tau(j|x,K) &= \sum_{k > K} \langle \phi_k, \xi^{'}_{\tau,x} \rangle \tilde{\gamma}_{j,k}^{1,1}(0)\\
\bbalpha^2_\tau(j|x,K) & = \sum_{k_1 , k_2 > K} \langle \phi_{k_1}, \xi^{'}_{\tau,x} \rangle \langle \phi_{k_2}, \xi^{'}_{\tau,x} \rangle \tilde{\gamma}_{j,k_1 ,k_2}^{2,1,1}(0)\\
\bbbeta_\tau^1(0|x)& =  \sum_{k_1 , k_2 > 0} \langle \xi_{\tau,x}^{''} \phi_{k_1}, \phi_{k_2} \rangle \tilde{\gamma}_{k_{1},k_{2}}^{1,1} (0) \\
\bbbeta_\tau^1(j|x)& =  \sum_{k_1 , k_2 > 0} \langle \xi_{\tau,x}^{''} \phi_{k_1}, \phi_{k_2} \rangle \tilde{\gamma}_{j,k_{1},k_{2}}^{1,1, 1} (0) \\
\bbbeta_\tau^2(0|x)& =  \sum_{k_1 , k_2, k_3, k_4 > 0} \langle \xi_{\tau,x}^{''} \phi_{k_1}, \phi_{k_2} \rangle \langle \xi_{\tau,x}^{''} \phi_{k_3}, \phi_{k_4} \rangle \tilde{\gamma}_{k_{1},k_{2}, k_{3},k_{4}}^{1,1, 1, 1} (0) \\
\bbbeta_\tau^2(j|x)& =  \sum_{k_1 , k_2, k_3, k_4 > 0} \langle \xi_{\tau,x}^{''} \phi_{k_1}, \phi_{k_2} \rangle \langle \xi_{\tau,x}^{''} \phi_{k_3}, \phi_{k_4} \rangle \tilde{\gamma}_{j,k_{1},k_{2}, k_{3},k_{4}}^{2,1,1, 1, 1} (0).
\end{align*}

\subsubsection*{Global bandwidth}
Lemma \ref{lemma:others} involves the terms
\begin{align*}
\delta_{\tau,0} & =  \frac{\sum_{i=1}^n  J_\tau (F(Y_i|x) L(h^{-1}||X_i - x||)}{n\mathbb{E}(J_\tau (F(Y_1|x)) L(h^{-1}||X_1 - x||))} \\
\delta_{\tau,j} & =  \frac{\sum_{i=1}^n \langle \phi_j , X_i - x\rangle J_\tau (F(Y_i|x) L(h^{-1}||X_i - x||)}{n\mathbb{E} (J_\tau (F(Y_1|x)) L(h^{-1}||X_1 - x||))} \\
[\Delta_\tau]_{j_1, j_2} & =  \frac{\sum_{i=1}^n \langle \phi_{j_1} , X_i - x\rangle \langle \phi_{j_2} , X_i - x\rangle J_\tau (F(Y_i|x) L(h^{-1}||X_i - x||)}{n\mathbb{E}(J_\tau (F(Y_1|x)) L(h^{-1}||X_1 - x||))} \\
A_{\tau,0} & =  \frac{\sum_{i=1}^n \langle \mathcal{P}^\perp_K \xi^{'}_{\tau,x} , X_i - x\rangle  L(h^{-1}||X_i - x||)}{n\mathbb{E} ( L(h^{-1}||X_1 - x||))} \\
A_{\tau,j} & =  \frac{\sum_{i=1}^n \langle \mathcal{P}^\perp_K \xi^{'}_{\tau,x} , X_i - x\rangle \langle \phi_j, X_i - x \rangle  L(h^{-1}||X_i - x||)}{n\mathbb{E} ( L(h^{-1}||X_1 - x||))} \\
B_{\tau,0} & =  \frac{\sum_{i=1}^n R_{\tau,x,x,i}  L(h^{-1}||X_i - x||)}{n\mathbb{E} ( L(h^{-1}||X_1 - x||))} \\
B_{\tau,j} & =  \frac{\sum_{i=1}^n R_{\tau,x,x,i} \langle \phi_j, X_i - x \rangle  L(h^{-1}||X_i - x||)}{n\mathbb{E} ( L(h^{-1}||X_1 - x||))} \\
C_{\tau,0} & =  \frac{\sum_{i=1}^n (R_{\tau,\rho,x,i} - R_{\tau,x,x,i} ) L(h^{-1}||X_i - x||)}{n\mathbb{E} ( L(h^{-1}||X_1 - x||))} \\
C_{\tau,j} & =  \frac{\sum_{i=1}^n (R_{\tau,\rho,x,i} - R_{\tau,x,x,i} )  \langle \phi_j, X_i - x \rangle  L(h^{-1}||X_i - x||)}{n\mathbb{E} ( L(h^{-1}||X_1 - x||))},
\end{align*}
and additional terms in proof on Theorem 3.2 are
\begin{align*}
\delta_0 & =  \frac{\sum_{i=1}^n   (L(h^{-1}||X_i - x||))^2}{(n\mathbb{E} (L(h^{-1}||X_1 - x||)))^2} \\
\delta_j & =  \frac{\sum_{i=1}^n \langle \phi_j , X_i - x\rangle  (L(h^{-1}||X_i - x||))^2}{(n\mathbb{E} (L(h^{-1}||X_1 - x||)))^2} \\
[\Delta]_{j_1, j_2} & =  \frac{\sum_{i=1}^n \langle \phi_{j_1} , X_i - x\rangle \langle \phi_{j_2} , X_i - x\rangle  (L(h^{-1}||X_i - x||))^2}{(n\mathbb{E} (L(h^{-1}||X_1 - x||)))^2}.
\end{align*}
\subsubsection*{Local bandwidth}
Lemma \ref{lemma:elements_knn} and Theorem 3.3 involve
\begin{align*}
\delta^\ell_{\tau,0} & =  \frac{\sum_{i=1}^n  J_\tau (F(Y_i|x) L((h^k_{i})^{-1}||X_i - x||)}{n\mathbb{E} (J_\tau (F(Y|x)) L((h_{1}^k)^{-1}||X_1 - x||))} \\
\delta^\ell_{\tau,j} & =  \frac{\sum_{i=1}^n \langle \phi_j , X_i - x\rangle J_\tau (F(Y_i|x) L((h^k_{i})^{-1}||X_i - x||)}{n\mathbb{E} (J_\tau (F(Y|x)) L((h_{1}^k)^{-1}||X_1 - x||))} \\
[\Delta^\ell_\tau]_{j_1, j_2} & =  \frac{\sum_{i=1}^n \langle \phi_{j_1} , X_i - x\rangle \langle \phi_{j_2} , X_i - x\rangle J_\tau (F(Y_i|x) L((h^k_{i})^{-1}||X_i - x||)}{n\mathbb{E} (J_\tau (F(Y|x)) L((h_{1}^k)^{-1}||X_1 - x||))} \\
A^\ell_{\tau,0} & =  \frac{\sum_{i=1}^n \langle \mathcal{P}^\perp_K \xi^{'}_{\tau,x} , X_i - x\rangle  L((h^k_{i})^{-1}||X_i - x||)}{n\mathbb{E} ( L((h_{1}^k)^{-1}||X_1 - x||))} \\
A^\ell_{\tau,j} & =  \frac{\sum_{i=1}^n \langle \mathcal{P}^\perp_K \xi^{'}_{\tau,x} , X_i - x\rangle \langle \phi_j, X_i - x \rangle  L((h^k_{i})^{-1}||X_i - x||)}{n\mathbb{E} (L((h_{1}^k)^{-1}||X_1 - x||))} \\
B^\ell_{\tau,0} & =  \frac{\sum_{i=1}^n R_{\tau,x,x,i}  L((h^k_{i})^{-1}||X_i - x||)}{n\mathbb{E} ( L((h_{1}^k)^{-1}||X_1 - x||))} \\
B^\ell_{\tau,j} & =  \frac{\sum_{i=1}^n R_{\tau,x,x,i} \langle \phi_j, X_i - x \rangle  L((h^k_{i})^{-1}||X_i - x||)}{n\mathbb{E} (L((h_{1}^k)^{-1}||X_1 - x||))}.
\end{align*}

\subsection{Assumptions}
\label{sec:assumptions}
\subsubsection*{Local linear scalar-on-function regression}
The following are the assumptions from \cite{ferraty22}, adapted to the local linear extremile regression.
\begin{enumerate}
\renewcommand{\labelenumi}{(A\arabic{enumi})}
\item For any $u$ in the neighborhood of $0\in E$, there exists $\nu = x + tu$, with $t\in (0,1)$, such that, for fixed $\tau \in (0,1)$
\begin{align*}
    \xi_\tau(x + u) = \xi_\tau(x) + \langle \xi_{\tau,x}^{'}, u \rangle + \frac{1}{2} \langle \xi_{\tau, \nu}^{''} u, u \rangle, 
\end{align*}
with extremile regression operator $\xi_\tau: E \to \mathbb{R} $, its functional derivative $\xi^{'}_{\tau,x}: E \to E$ and $\xi^{''}_{\tau,\nu}: E \to E$ is a Hilbert--Schmidt linear operator such that there exists $C \leq 0$ for which $d(\xi^{''}_{\tau,\nu_1}, \xi^{''}_{\tau,\nu_2}) \leq C d(\nu_1, \nu_2)$ for any 
$\nu_1(\cdot), \nu_2(\cdot)$ in the neighborhood of $x(\cdot) \in E$.
\item The kernel function $L$ is continuously differentiable on its support $(0,1)$ with $L'(s)\leq 0$ for all $s \in (0,1)$ and $L(1) >0$.
    \item The functions $\gamma_{j_1}^1, \gamma_{j_{1}, j_{2}}^{1, 1}, \dots, \gamma_{j_{1}, j_{2} , j_{3}, j_{4}}^{1, 1, 1, 1}$,  $\gamma_{j_1}^2, \gamma_{j_{1}, j_{2} , j_{3}}^{2, 1, 1}, \dots, \gamma_{j_{1}, j_{2} , j_{3}, j_{4}, j_{5}}^{2, 1, 1, 1, 1}$, $\gamma_{j_{1}, j_{2}}^{2, 2}$ are continuously differentiable around 0. The smallest eigenvalue of $\Gamma$, denoted $\lambda_K$, is strictly positive.
  \item The bandwidth sequence $h=h_n$ is such that as $n\to\infty$ we have $h_n \to 0$; The number of basis functions used $K=K_n$ is such that $K_n\to \infty$ and $n\pi_x(h)\to \infty$. Finally, $h\sqrt{K} = \smallo(1)$, $\frac{1}{h \lambda_K \sqrt{n\pi_x(h)}} = \smallo(1)$ and $\frac{\sqrt{K}}{h \sqrt{\lambda_K n\pi_x(h)}} = \smallo(1)$.
      \item For any $s>0$, $\frac{\pi_x(sh)}{\pi_x(h)} \to \theta_x(s)$ as $h\to 0$.
    \item For fixed $\tau \in (0,1)$, operator $\sigma^2_\tau(x) = \text{Var}(J_\tau(F(Y|X))Y|X=x)$ is uniformly continuous.
    \item $h \mapsto \pi_x(h)$ is a continuous function in the neighborhood of 0 such that $\theta_x(s) \geq s$ and $\theta_x(s^{-1})\leq s^{-1}$ $\forall s >1$; moreover, the number of neighbours in the kNN algorithm $k=k_n$ is such that $\frac{k}{n}=\smallo(1)$, $\frac{\sqrt{K}}{\pi_x\left(\frac{k}{n}\right)}=\smallo(1)$,
    $\frac{1}{\pi_x\left(\frac{k}{n}\right)\lambda_K \sqrt{k}}=\smallo(1)$, and 
$\frac{\sqrt{\lambda_K}}{\pi_x\left(\frac{k}{n}\right)\sqrt{K k}}=\smallo(1)$,
 with $K = K_n \to \infty$ as $n\to \infty$.
\end{enumerate}

\subsubsection*{Uniform consistency of conditional cumulative distribution function estimator}
We rely on the assumption from \cite{Ferraty10}, and we re-state them here for completeness.

Define the subspaces $\mathcal{S}_x \subset E$ and $\mathcal{S}_y \subset \mathbb{R}$, and the open ball $ B_r(x) = \{\tilde{x} \in E : d(x, \tilde{x}) \leq r \}$. Two functions are key components in this framework: $\varphi$, controlling the concentration of probability measure of functional random variable $X$ on $B_r(x)$, and $\psi_{\mathcal{S}}(\epsilon)$, Kolmogorov's $\epsilon$-entropy of a set $\mathcal{S} \subset E$. In particular, the latter is defined as the minimal number of open balls of radius $\epsilon$ needed to cover $\mathcal{S}$. Intuitively, the lower $\psi_{\mathcal{S}}(\epsilon)$, the less complex the topological structure of $\mathcal{S}$.

Now consider kernel $L_F$ with bandwidth $h_F$ such that $h_F = h_{F,n} \to 0$ with $n \to \infty$. Consider the following assumptions.
\begin{enumerate}[label=(A\arabic*)] 
\setcounter{enumi}{7}

\item For all $(y_1, y_2) \in \mathcal{S}_y  \times \mathcal{S}_y$, and all  $(x_1, x_2) \in \mathcal{S}_x  \times \mathcal{S}_x$ 
$$|F(y_1|x_1) - F(y_2|x_2)| \leq C\left(d(x_1,x_2)^{b_x} + |y_1 - y_2|^{b_y} \right),$$
with $C, b_x, b_y>0$.

\item The functions $\varphi$ and $\psi_{\mathcal{S}_x}$ are such that
\begin{itemize}
    \item There exist $ C>0$, and $\eta_0 >0$, s.t. for all $ \eta < \eta_0$ 
    $$\frac{\text{d}}{\text{d}t} \varphi (t)\biggr\rvert_{\eta} <C.$$ \\
    Additionally, if $L^F(1)=0$: there exist $C>0$, and $ \eta_0 >0$, s.t. for all $ 0<\eta < \eta_0 $
    $$\int_0^\eta \varphi(u) \text{d}u > C\eta \varphi(\eta).$$
    \item For $n$ large enough:
    $$ \frac{(\log n)^2}{n \varphi(h_F)} < \psi_{\mathcal{S}_x} \left(\frac{\log n}{n}\right) < \frac{n \varphi(h_F)}{\log n}.$$
\end{itemize}

 \item $\forall x \in \mathcal{S}_x$ 
$$0 < C \varphi(h) \leq \Pr(X \in B_r(x)) \leq C^{'} \varphi(h) < \infty, $$
with $C,C^{'} >0$.

\item $L_F$ is a bounded and Lipschitz kernel on its support $[0,1)$, and if $L_F(1)=0$, the kernel $L_F(1)$ fulfills the additional condition $-\infty < C < \frac{\text{d}}{\text{d}t}L_F(t) < C^{'} < 0$, with $C,C^{'} >0$.

\item For some $\beta >1$, Kolmogorov's $\varepsilon$-entropy of $\mathcal{S}_x$ satisfies
\begin{align*}
    \sum_{n=1}^\infty n^\frac{1}{2b_y} e^{(1-\beta)\psi_{\mathcal{S}_x}\left(\frac{\log n}{n}\right)} < \infty.
\end{align*}
\end{enumerate}

It was shown in \cite{Ferraty10} that, under assumptions (A8) -- (A12), the estimator 
\begin{align*}
    \label{eq:F_hat_fct}
    \hat{F}_{h_{F}}(y|x) = \frac{\sum_{i=1}^n L_F \left(h_F^{-1}||x-X_i||\right) {1}(Y_i \leq y) }{\sum_{i=1}^n L_F \left(h_F^{-1}||x-X_i||\right)}, \hspace{5mm} \forall y\in \mathbb{R}, \forall x\in E,
\end{align*}
is uniformly consistent. In particular, for $\mathcal{S}_x \subset E$ and $\mathcal{S}_y \subset \mathbb{R}$
\begin{align*}
    \sup_{x \in \mathcal{S}_x } \sup_{y \in  \mathcal{S}_y }|F(y|x) - \hat{F}_{h_{F}}(y|x) | \to 0 \text{ a.s.}
\end{align*}


\subsection{Technical lemmas}
This section is dedicated to stating and proving the technical lemmas involved in the proofs in Section \ref{sec:proof_main_thms}. Notice that Lemmas \ref{lemma:S8} and \ref{lemma:S8_kNN} are pivotal to generalize the local linear scalar-on-functional regression to the weighted version we propose. The rest of the results follow from the said lemmas and further results in the supplementary of \cite{ferraty22}. Finally, Lemma \ref{lemma:kNN_pi} corresponds to Lemma S.5 of \cite{ferraty22} and we add it here for completeness.

\subsubsection*{Global bandwidth}

\begin{lemma}
\label{lemma:S8}
For $i = 1,\dots,n$, $p\geq0$, $q>0$, and $\tau \in (0,1)$, we have
$$    \int_\mathbb{R}\int_0^1 t^p\left[L(t)J_{\tau}(\hat{F}_{h_{F}}(y | x))\right]^q  \text{d}F_{||X_i - x ||/h}(t) \text{d}F_{Y}(y) = C_{\tau}(p,q|x) \pi_x(h)(1 + \smallo(1)),$$
where $C_{\tau}(p,q|x)=\mathbb{E}(J_\tau(F(Y|x))^q) \left\{ L^q(1)-\int_0^1 \frac{\de}{\text{d}u}(u^p L^q(u))\vartheta_x(u)\text{d}u\right \}$.

\end{lemma}
\begin{proof}
For $\tau \in (0,1/2]$, define 
$ R_{s(\tau)}(y) = \frac{(1-\hat{F}_{h_{F}}(y | x))^{s(\tau)-1}}{(1-F(y | x))^{s(\tau)-1}}. $
Then $R_{s(\tau)} \to 1$ uniformly almost surely by Theorem 4 of \cite{Ferraty10} and so it may be seen that using the convergence result of Lemma S8 in \cite{ferraty22} we readily obtain

\begin{align*}
    &   \int_\mathbb{R}\int_0^1 t^p\left[L(t)s(\tau)(1-\hat{F}_{h_{F}}(y | x))^{s(\tau)-1}\right]^q  \text{d}F_{||X_i - x ||/h}(t) \text{d}F_{Y}(y) \\
   &= s(\tau)^q \left [ \int_0^1 t^p L(t)^q \text{d}F_{||X_i - x ||/h}(t) \right ] \int_\mathbb{R}(1-\hat{F}_{h_{F}}(y | x))^{q(s(\tau)-1)} \text{d}F_{Y}(y) \\
   & = s(\tau)^q \left [\int_0^1 t^p L(t)^q \text{d}F_{||X_i - x ||/h}(t) \right ] \int_\mathbb{R} \left(R_{s(\tau)}(y) (1-F(y | x))\right)^{q(s(\tau)-1)} \text{d}F_{Y}(y) \\
   & = s(\tau)^q \left [ \int_0^1 t^p L(t)^q \text{d}F_{||X_i - x ||/h}(t) \right ] \mathbb{E} \left[ \left(R_{s(\tau)}(y) (1-F(y | x))\right)^{q(s(\tau)-1)} \right ] \\
   & = s(\tau)^q \left [ \int_0^1 t^p L(t)^q \text{d}F_{||X_i - x ||/h}(t) \right ]  \mathbb{E} \left[ \left( (1-F(y | x))\right)^{q(s(\tau)-1)} \right ](1 + \smallo(1)) \\
   & = C_{\tau}(p,q|x) \pi_x(h)(1 + \smallo(1)),
\end{align*}
where  
$$C_{\tau}(p,q|x) =   \mathbb{E} \left[s(\tau)^q  (1-F(Y | x))^{q(s(\tau)-1)} \right ]\left\{ L^q(1)-\int_0^1 \frac{\de}{\text{d}u}(u^p L^q(u))\vartheta_x(u)\text{d}u\right \}.$$
With similar arguments, the result holds for $\tau \in [1/2,1)$ with the constant 
$$C_{\tau}(p,q|x) =   \mathbb{E} \left[r(\tau)^q  F(Y | x)^{q(r(\tau)-1)} \right ]\left\{ L^q(1)-\int_0^1 \frac{\de}{\text{d}u}(u^p L^q(u))\vartheta_x(u)\text{d}u\right \}.$$ Collecting the constants proves the result. 

\end{proof}

\begin{corollary}
\label{cor:moments_L}
For any $q>0$
\begin{align*}
    \mathbb{E}((L(h^{-1}||x-X ||))^q) = C_{0.5}(x|0,q)\pi_x(h)\{1 + \smallo(1) \}.
\end{align*}
\end{corollary}
\begin{proof}
The result follows from $J_{0.5}(F(y|x))=1$.
\end{proof}

\begin{lemma}
\label{lemma:S9}
  We have that
  \begin{enumerate}
      \item For any $q>0$ 
      \begin{align*}
      &\mathbb{E}((L(h^{-1}(x-X)))^q J_{\tau}(\hat{F}_{h_{F}}(Y | x))\gamma_{j_1,\dots,j_K}^{p_1,\dots,p_K}(||X - x ||^p))\\
      &=h^p\, C_{\tau}(p,q|x)\, \tilde\gamma_{j_1,\dots,j_K}^{p_1,\dots,p_K}(0) \pi_x(h)(1 + \smallo(1)),
    \end{align*}
    with $p = \sum_{k=1}^K p_{k}$.
    \item All the following sums are smaller or equal to $1$:
    \begin{align*}
        &\sum_j (\tilde{\gamma}_j^1(0))^2,\quad
        \sum_{j_1,j_2} (\tilde{\gamma}_{j_1,j_2}^{1,1}(0))^2,\quad
        \sum_{j_1,j_2,j_3} (\tilde{\gamma}_{j_1,j_2, j_3}^{1,1,1}(0))^2,\\
        &\sum_{j_1,j_2,j_3, j_4} (\tilde{\gamma}_{j_1,j_2, j_3, j_4}^{1,1,1,1}(0))^2, \quad
        \sum_j \tilde{\gamma}_j^2(0), \quad
        \sum_{j_1,j_2} \tilde{\gamma}_{j_1,j_2}^{2,2}(0).
    \end{align*}
    \item $\Gamma$ is positive semidefinite.
\end{enumerate}
    
\end{lemma}

\begin{proof}
Using the adapted Lemma \ref{lemma:S8}, the rest of the proof is analogous to \cite[Lemma S.9]{ferraty22}.
\end{proof}

\begin{corollary}
\label{cor:upper_bound_alpha}
The following holds\begin{align*}
    &| \bbalpha^1_\tau(0|x,K)|  \leq ||\ProjK^\perp \xi_\tau^{'} ||,\quad
    \sum_{k=1}^K (\bbalpha^1_\tau(k|x,K))^2 \leq     ||\ProjK^\perp \xi_\tau^{'} ||^2,\quad
     \sum_{k=1}^K \bbalpha^2_\tau(k|x,K) \leq     ||\ProjK^\perp \xi_\tau^{'} ||^2\\
     &\bbbeta^1_\tau(0,x)  = \bigo(1),\quad
     \sum_{k=1}^K (\bbbeta^1_\tau(k,x))^2  = \bigo(1),\quad
     \sum_{k=1}^K \bbbeta^2_\tau(k,x)  = \bigo(1).
\end{align*}
\end{corollary}
\begin{proof}
Follows from Lemma \ref{lemma:S9} and the Cauchy-Schwartz inequality.
\end{proof}

\begin{lemma}
\label{lemma:others}
  Under the hypothesis (A1) -- (A5)
  \begin{enumerate}
      \item For all $j, j_1, j_2 \in \{1,\dots,K\}$ 
       \begin{align*}
       \delta_{\tau,j} &= \frac{C_\tau(1,1|x)}{C_\tau(0,1|x)} \tilde{\gamma}_j^1(0)h\{1 + \smallo_{\Pr} (1)\} + \bigo_{\Pr}(h \{n \pi_x(h) \}^{-1/2})\sqrt{\tilde{\gamma}^2_j(0)}\\
       [\Delta_\tau]_{j_1, j_2} & = \frac{C_\tau(2,1|x)}{C_\tau(0,1|x)} \tilde{\gamma}_{j_1, j_2}^{1,1}(0)h^2\{1 + \smallo_{\Pr} (1)\} + \bigo_{\Pr}(h^2 \{n \pi_x(h) \}^{-1/2})\sqrt{\tilde{\gamma}_{j_1, _2}^{2,2}(0)}\\
    \end{align*}
      \item For all $j\in \{1,\dots,K\}$ 
      \begin{align*}
         A_{\tau,0} & = \frac{C_\tau(1,1|x)}{C_\tau(0,1|x)} \bbalpha_\tau^1(0|x,K) h\{1 + \smallo (1)\} + \bigo_{\Pr}(||\ProjK^\perp \xi^{'}_{\tau,x} || h \{n \pi_x(h) \}^{-1/2})\\
        A_{\tau,j} & = \frac{C_\tau(2,1|x)}{C_\tau(0,1|x)} \bbalpha_\tau^1(j|x,K)h^2\{1 + \smallo (1)\} + \bigo_{\Pr}(\sqrt{\bbalpha_\tau^2(j|x,K)} h^2 \{n \pi_x(h) \}^{-1/2})
      \end{align*}
      \item
      For any $\mathbf{u},\mathbf{v} \in \mathbb{R}^K $, $| \mathbf{u}^{\top} \Gamma^{-1}\mathbf{v}| \leq \lambda_K^{-1} ||\mathbf{u} ||_2 ||\mathbf{v} ||_2 $, where $ \lambda_K $ is the smallest eigenvalue of $\Gamma$. 
       
      \item For all $j\in \{1,\dots,K\}$ 
   \begin{align*}
         B_{\tau,0} & = \frac{C_\tau(2,1|x)}{C_\tau(0,1|x)} \bbbeta_\tau^1(0|x) h^2\{1 + \smallo (1)\} + \bigo_{\Pr}(h^2\{n\pi_x(h)\}^{-1/2})\\
        B_{\tau,j} & = \frac{C_\tau(3,1|x)}{C_\tau(0,1|x)}  \bbbeta_\tau^1(j|x)h^3\{1 + \smallo (1)\} + \sqrt{\bbbeta^2_\tau(j|x)} \bigo_{\Pr}(h^3\{n\pi_x(h)\}^{-1/2})
      \end{align*}

  \end{enumerate}
    
\end{lemma}

\begin{proof}
Using Lemmas \ref{lemma:S8} and \ref{lemma:S9}, the rest of the proof is analogous to \cite[Lemmas S.10 -- S.14]{ferraty22}.
\end{proof}

\begin{lemma}
\label{lemma:elements_var}
Under the hypothesis (A1) -- (A5) and using Corollary \ref{cor:moments_L}, for all $j, j_1, j_2\in \{1,\dots,K\}$ 
       \begin{align*}
       \delta_0& = \frac{C_{0.5}(0,2|x)}{(C_{0.5}(0,1|x))^2}\{n \pi_x(h) \}^{-1}\{1 + \smallo_{\Pr} (1)\}\\
       \delta_j &= \frac{C_{0.5}(1,2|x)}{(C_{0.5}(0,1|x))^2} \tilde{\gamma}_j^1(0)h \{n \pi_x(h) \}^{-1} \{1 + \smallo_{\Pr} (1)\} + \bigo_{\Pr}(h \{n \pi_x(h) \}^{-3/2})\sqrt{\tilde{\gamma}^2_j(0)}\\
       [\Delta]_{j_1, j_2} & = \frac{C_{0.5}(2,2|x)}{(C_{0.5}(0,1|x))^2} \tilde{\gamma}_{j_1, j_2}^{1,1}(0)h^2  \{n \pi_x(h) \}^{-1} \{1 + \smallo_{\Pr} (1)\} + \bigo_{\Pr}(h^2 \{n \pi_x(h) \}^{-3/2})\sqrt{\tilde{\gamma}_{j_1, _2}^{2,2}(0)}.
    \end{align*}
\end{lemma}
\begin{proof}
    Analogous to proof of 1. of Lemma \ref{lemma:others}.
\end{proof}

\subsubsection*{Local bandwidth}
\begin{lemma}
\label{lemma:kNN_pi}
Under (A7), a.s.
\begin{align*}
    \pi_x(h^k) \sim \frac{k}{n},\quad h^k \sim \pi^{-1}_x\left( \frac{k}{n}\right),
\end{align*}
where $h^k= \inf \bigg \{h \in H: \sum_{j=1}^n \boldsymbol{1}_{X_j \in B_h(x)} = k \bigg \}$.
\end{lemma}
\begin{proof}
    See proof of Lemma S.5 of supplementary materials to \cite{ferraty22}.
\end{proof}

\begin{lemma}
\label{lemma:S8_kNN}
Under (A2), (A3), (A5) and (A7), for $i \in \{ 1,\dots,n\}$, $p\geq0$, $q>0$, and $\tau \in (0,1)$, we have
$$    \int_\mathbb{R}\int_0^1 t^p\left[L(t)J_{\tau}(\hat{F}_{h_{F}}(y | x))\right]^q  \text{d}F_{||X_i - x ||/h^k_{i}}(t) \text{d}F_{Y}(y) =  C_{\tau}(p,q|x) \frac{k}{n}(1 + \smallo(1))$$
where $C_{\tau}(p,q|x)$ is defined in Lemma \ref{lemma:S8}.
\end{lemma}

\begin{proof}
 Consider $\mathcal{S}_{-i} = \{X_1,\dots,X_{i-1},X_{i+1}, \dots,  X_n \}$ and write
\begin{align*}
    &\mathbb{E}(((h^k_{i})^{-1}||X_i - x ||)^p L^q((h^k_{i})^{-1}||X_i - x ||) | \mathcal{S}_{-i} )\\
    & = L^q(1)-\int_0^1 \frac{\de}{\text{d}u}(u^p L^q(u))\pi_x(h^k_{i}u)\text{d}u\\
    & = \pi_x(h^k_{i})\left\{ L^q(1)-\int_0^1 \frac{\de}{\text{d}u}(u^p L^q(u)) \frac{\pi_x(h^k_{i}u)}{\pi_x(h^k_{i})}\text{d}u\right \}\\
    & \sim \pi_x(h^k_{i})\left\{ L^q(1)-\int_0^1 \frac{\de}{\text{d}u}(u^p L^q(u)) \theta_x(u)\text{d}u\right \} \text{ a.s}, 
\end{align*}
where the last step comes from the fact that $\sup_{s\in[0,1]} \left\lvert\frac{\pi_x(h^ks)}{\pi_x(h^k)} - \theta_x(s) \right\rvert\to 0$ almost surely (see proof of Lemma S.6-(i) in the supplementary of \cite{ferraty22}). Thus, from Lemma \ref{lemma:kNN_pi} and the law of total expectation, we have
\begin{align*}
     &\mathbb{E}(((h^k_{i})^{-1}||X_i - x ||)^p L^q((h^k_{i})^{-1}||X_i - x ||) | \mathcal{S}_{-i} )\\
     &= \frac{k}{n}\left\{ L^q(1)-\int_0^1 \frac{\de}{\text{d}u}(u^p L^q(u)) \theta_x(u)\text{d}u\right \}(1 + \smallo(1)).
\end{align*}
Hence, similarly to the proof of Lemma \ref{lemma:S8}, we obtain the result
\begin{align*}
   &\int_\mathbb{R}\int_0^1 t^p\left[L(t)J_{\tau}(\hat{F}_{h_{F}}(y | x))\right]^q  \text{d}F_{||X_i - x ||/h^k_{i}}(t) \text{d}F_{Y}(y) \\
  & =  C_{\tau}(p,q|x) \frac{k}{n}(1 + \smallo(1)).
\end{align*}
\end{proof}
\begin{lemma}
\label{lemma:S9_kNN}
    Under (A2), (A3), (A5) and (A7), for $i \in \{ 1,\dots,n \}$, for any $q>0$, and $\tau \in (0,1)$, we have
    \begin{align*}
      &\mathbb{E}(L^q((h^k_{i})^{-1}(x-X)) J_{\tau}(\hat{F}_{h_{F}}(Y | x))\gamma_{j_1,\dots,j_K}^{p_1,\dots,p_K}(||X - x ||^p))\\
      &=\left(\pi_x^{-1}\left(\frac{k}{n}\right)\right)^p\, C_{\tau}(p,q|x)\, \tilde\gamma_{j_1,\dots,j_K}^{p_1,\dots,p_K}(0) \frac{k}{n}(1 + \smallo(1))
    \end{align*}
    with $p = \sum_{k=1}^K p_{k}$.
\end{lemma}

\begin{proof}
    Using Lemmas \ref{lemma:kNN_pi} and \ref{lemma:S8_kNN}, the rest of the proof is analogous to \cite[Lemma S.9]{ferraty22}.
\end{proof}

\begin{lemma}
\label{lemma:elements_knn}
Under (A2), (A3), (A5) and (A7), for all $j, j_1, j_2=1,\dots,K$ we have
 \begin{align*}
       \delta^\ell_{\tau,j} &= \frac{C_\tau(1,1|x)}{C_\tau(0,1|x)} \tilde{\gamma}_j^1(0)\pi_x^{-1}\left(\frac{k}{n}\right)\{1 + \smallo_{\Pr} (1)\} + \bigo_{\Pr}\left(\pi_x^{-1}\left(\frac{k}{n}\right) k ^{-1/2}\right)\sqrt{\tilde{\gamma}^2_j(0)}\\
       [\Delta^\ell_\tau]_{j_1, j_2} & = \frac{C_\tau(2,1|x)}{C_\tau(0,1|x)} \tilde{\gamma}_{j_1, j_2}^{1,1}(0)\left(\pi_x^{-1}\left(\frac{k}{n}\right)\right)^2\{1 + \smallo_{\Pr} (1)\} + \bigo_{\Pr}\left(\left(\pi_x^{-1}\left(\frac{k}{n}\right)\right)^2k ^{-1/2}\right)\sqrt{\tilde{\gamma}_{j_1, _2}^{2,2}(0)}\\
        A^\ell_{\tau,0} & = \frac{C_\tau(1,1|x)}{C_\tau(0,1|x)} \bbalpha_\tau^1(0|x,K) \pi_x^{-1}\left(\frac{k}{n}\right)\{1 + \smallo (1)\} + \bigo_{\Pr}\left(||\ProjK^\perp \xi^{'}_{\tau,x} ||\pi_x^{-1}\left(\frac{k}{n}\right) k^{-1/2}\right)\\
        A^\ell_{\tau,j} & = \frac{C_\tau(2,1|x)}{C_\tau(0,1|x)} \bbalpha_\tau^1(j|x,K)\left(\pi_x^{-1}\left(\frac{k}{n}\right)\right)^2\{1 + \smallo (1)\} + \bigo_{\Pr}\left(\sqrt{\bbalpha_\tau^2(j|x,K)} \left(\pi_x^{-1}\left(\frac{k}{n}\right)\right)^2 k^{-1/2}\right)\\
         B^\ell_{\tau,0} & = \frac{C_\tau(2,1|x)}{C_\tau(0,1|x)} \bbbeta_\tau^1(0|x) \left(\pi_x^{-1}\left(\frac{k}{n}\right)\right)^2\{1 + \smallo (1)\} + \bigo_{\Pr}\left(\left(\pi_x^{-1}\left(\frac{k}{n}\right)\right)^2 k^{-1/2}\right)\\
        B^\ell_{\tau,j} & = \frac{C_\tau(3,1|x)}{C_\tau(0,1|x)}  \bbbeta_\tau^1(j|x)\left(\pi_x^{-1}\left(\frac{k}{n}\right)\right)^3\{1 + \smallo (1)\} + \sqrt{\bbbeta^2_\tau(j|x)} \bigo_{\Pr}\left(\left(\pi_x^{-1}\left(\frac{k}{n}\right)\right)^3 k^{-1/2}\right).
    \end{align*}
\end{lemma}
\begin{proof}
Using Lemmas \ref{lemma:S8_kNN} and \ref{lemma:S9_kNN}, the rest of the proof is analogous to \cite[Lemma S.10, S.11, S.12, S.14]{ferraty22}.
\end{proof}

\subsection{Proofs of main theorems}
\label{sec:proof_main_thms}
In this section we extend the proofs by \cite{ferraty22} for local linear scalar-on-function regression to the extremile scalar-on-function regression. In particular, we provide the proofs for the asymptotic bias and variance expressions of the conditional extremile estimator.

\begin{proof}[Proof of Theorem 3.1]
    We require checking the order of magnitude of (17) and (18) in the manuscript. 
\begin{enumerate}
    \item First, we rewrite $T_1$ as
$$T_1 =  \boldsymbol{e}^{\top}({\boldsymbol{\Phi}}^{\top} \widetilde{\mathbf{W}}_\tau \boldsymbol{\Phi})^{-1}(A_{\tau,0},A_{\tau,1},\dots,A_{\tau,K})^{\top}$$
with $A_{\tau,0}$ and $A_{\tau,j}$, $j \in \{1,\dots,K\}$, defined in Section \ref{sec:notation}.
The matrix to invert is ${\boldsymbol{\Phi}}^{\top} \widetilde{\mathbf{W}}_\tau \boldsymbol{\Phi}$, where $ \widetilde{\mathbf{W}}_\tau = (n \mathbb{E}(J_\tau(F(Y|x)){L(h^{-1}(x-X))}))^{-1} \mathbf{W}_\tau$. Hence we have that
\begin{equation}
\label{eq:mat_to_inv}
    {\boldsymbol{\Phi}}^{\top} \widetilde{\mathbf{W}}_\tau \boldsymbol{\Phi} = \begin{bmatrix}
\delta_{\tau,0}& \delta^{\top}_{\tau}\\
\delta_\tau & \Delta_\tau
\end{bmatrix},
\end{equation}
with $\delta_\tau = (\delta_{\tau,1},\dots,\delta_{\tau,K})^{\top}$ and $\Delta$ denoting the $K \times K$ matrix whose entries are defined in Section \ref{sec:notation}.
With the inversion of the block matrix \eqref{eq:mat_to_inv} we get
\begin{equation}
\label{eq:inv_mat}
   \left ( {\boldsymbol{\Phi}}^{\top} \widetilde{\mathbf{W}}_\tau \boldsymbol{\Phi} \right)^{-1} = \begin{bmatrix}
(\delta_{\tau,0} - \delta_\tau^{\top} \Delta_\tau^{-1}\delta_\tau)^{-1} & -(\delta_{\tau,0} - \delta_\tau^{\top} \Delta_\tau^{-1}\delta_\tau)^{-1}\delta_\tau^{\top}\Delta_\tau^{-1}\\
-(\delta_{\tau,0} - \delta_\tau^{\top} \Delta_\tau^{-1}\delta_\tau)^{-1} \Delta_\tau^{-1} \delta_\tau & \Delta_\tau^{-1} + (\delta_{\tau,0} - \delta_\tau^{\top} \Delta_\tau^{-1}\delta_\tau)^{-1}\Delta_\tau^{-1}\delta_\tau \delta_\tau^{\top} \Delta_\tau^{-1}
\end{bmatrix},
\end{equation}
and hence
\begin{align*}
    T_1 = (\delta_{\tau,0} - \delta_\tau^{\top} \Delta_\tau^{-1}\delta_\tau)^{-1}(A_{\tau,0} - \delta_\tau^{T} \Delta_\tau^{-1} (A_{\tau,1},\dots,A_{\tau,K})^{\top}).
\end{align*}
One can see arguing as in \citep[Lemma S.1 (i)]{ferraty22} and using Lemma \ref{lemma:others} that
\begin{equation}
\label{eq:Delta_inv}
    \Delta_\tau^{-1} = \frac{C_\tau(0,1|x)}{C_\tau(2,1|x)} h^{-2} \Gamma^{-1}(1 + \smallo_{\Pr}(1)),
\end{equation}
where $\Gamma$ is a $K \times K$ defined in  Section \ref{sec:notation}. Then, as in \citep[Lemma S.1 (i)]{ferraty22} and using Lemmas \ref{lemma:S8}, \ref{lemma:S9} and \ref{lemma:others}, as well as the upper bounds in Corollary \ref{cor:upper_bound_alpha}, we see that
\begin{align}
\label{eq:conv_rate_mu}
    (\delta_{\tau,0} - \delta_\tau^{\top} \Delta_\tau^{-1}\delta_\tau)^{-1}& = \bigo_{\Pr}(\lambda_K)\\
    \nonumber
    (A_{\tau,0} - \delta_\tau^{T} \Delta_\tau^{-1} (A_{\tau,1},\dots,A_{\tau,K})^{\top}) & = \bigo_{\Pr}(h ||\ProjK^\perp \xi_\tau^{'}||) + \bigo_{\Pr}\left( \frac{h}{\sqrt{n \pi_x(h)}}||\ProjK^\perp \xi_\tau^{'}||\right)\\
    \nonumber
    &- \bigo_{\Pr}(\lambda_K^{-1} h ||\ProjK^\perp \xi_\tau^{'}||) + \bigo_{\Pr}\left(\lambda_K^{-1} \frac{h}{\sqrt{n \pi_x(h)}}||\ProjK^\perp \xi_\tau^{'}||\right) .
\end{align}
Hence, $T_1 = \bigo_{\Pr}(h ||\ProjK^\perp \xi_\tau^{'}||)$.

\item We rewrite $T_2$ as
\begin{align*}
    T_2 = \boldsymbol{e}^{\top}({\boldsymbol{\Phi}}^{\top} \widetilde{\mathbf{W}}_\tau \boldsymbol{\Phi})^{-1}\left[(B_{\tau,0},B_{\tau,1},\dots,B_{\tau,K})^{\top} + (C_{\tau,0},C_{\tau,1},\dots,C_{\tau,K})^{\top} \right],
\end{align*}
and once again we use inverse matrix \eqref{eq:inv_mat} to obtain 
\begin{align*}
    T_2 &= (\delta_{\tau,0} - \delta_\tau^{\top} \Delta_\tau^{-1}\delta_\tau)^{-1}(B_{\tau,0} - \delta_\tau^{T} \Delta_\tau^{-1} (B_{\tau,1},\dots,B_{\tau,K})^{\top}) \\
    &+ (\delta_{\tau,0} - \delta_\tau^{\top} \Delta_\tau^{-1}\delta_\tau)^{-1}(C_{\tau,0} - \delta_\tau^{T} \Delta_\tau^{-1} (C_{\tau,1},\dots,C_{\tau,K})^{\top}).
\end{align*}
Following similar arguments arguments to the ones used for $T_1$, with Lemmas \ref{lemma:S8}, \ref{lemma:S9}, \ref{lemma:others}, and Corollary \ref{cor:upper_bound_alpha}, we have that 
\begin{align*}
   &(B_{\tau,0} - \delta_\tau^{T} \Delta_\tau^{-1} (B_{\tau,1},\dots,B_{\tau,K})^{\top})\\
   & = \bigo_{\Pr}(h^2) + \bigo_{\Pr}\left(\frac{h^2}{\sqrt{n\pi_x(h)}} \right) - \bigo_{\Pr}(\lambda_K^{-1} h^2) + \bigo_{\Pr}\left(\lambda_K^{-1} \frac{h^2}{\sqrt{n\pi_x(h)}} \right). 
\end{align*}
Moreover, thanks to (A1), $C_{\tau,0} = \bigo_{\Pr}(h^3)$ and $C_{\tau,j} = \bigo_{\Pr}(h^4)$, for $j \in \{1,\dots,K\}$, hence $(C_{\tau,0} - \delta_\tau^{T} \Delta_\tau^{-1} (C_{\tau,1},\dots,C_{\tau,K})^{\top}) = (1 - h\delta_\tau^{\top} \Delta_\tau^{-1}(1,\dots,1)^{\top})\bigo_{\Pr}(h^3)$. 
Thanks to \eqref{eq:Delta_inv} we have that 
$$ \delta_\tau^{\top} \Delta_\tau^{-1}(1,\dots,1)^{\top} \propto \frac{1}{h^2} \delta_\tau^{\top} \Gamma^{-1}(1,\dots,1)^{\top},$$ and thanks to point 3. of Lemma \ref{lemma:others}, we obtain
\begin{align*}
    | \delta_\tau^{\top} \Gamma^{-1}(1,\dots,1)^{\top} | \leq \lambda_K^{-1} \sqrt{K} || \delta_\tau||_2.
\end{align*}
Since $|| \delta_\tau||_2 = \bigo_{\Pr}(h)$ because of point 1. Lemma \ref{lemma:others}, we finally obtain
\begin{align*}
    (C_{\tau,0} - \delta_\tau^{T} \Delta_\tau^{-1} (C_{\tau,1},\dots,C_{\tau,K})^{\top}) = \bigo_{\Pr}(\lambda_K^{-1} \sqrt{K} h^3).
\end{align*}
Using assumption $h\sqrt{K} = \smallo(1)$ in (A4), we finally get that $T_2 = \bigo_{\Pr}(h^2) + \smallo_{\Pr}(h^2)$.
\end{enumerate}
The two above points prove the desired result.
\end{proof}

\begin{proof}[Proof of Theorem 3.2]
    First consider 
    \begin{equation}
\label{eq:mat_in_var}
    {\boldsymbol{\Phi}}^{\top} \widetilde{\mathbf{W}^2} \boldsymbol{\Phi} = \begin{bmatrix}
\delta_{0}& \delta^{\top}\\
\delta & \Delta
\end{bmatrix},
\end{equation}
    where $\delta = (\delta_1,\dots, \delta_K)^{\top}$ and $\Delta$ denoting the $K \times K$ matrix whose entries are defined in Section \ref{sec:notation}. We can then rewrite (20) in the manuscript as
\begin{align}
    &\textup{Var}(\hat{\alpha}_\tau|\boldsymbol{X})\\
    & = [\sigma^2_\tau(x) + \smallo(1)](\delta_{\tau,0} - \delta_\tau^{\top} \Delta_\tau^{-1}\delta_\tau)^{-2}(\delta_0 - \delta_\tau^{\top}\Delta_\tau^{-1}\delta - \delta^{\top}\Delta_\tau^{-1}\delta_\tau + \delta_\tau^{\top}\Delta_\tau^{-1} \Delta \Delta_\tau^{-1}\delta_\tau).\nonumber
\end{align}
From \eqref{eq:conv_rate_mu} we know that 
\begin{align*}
    (\delta_{\tau,0} - \delta_\tau^{\top} \Delta_\tau^{-1}\delta_\tau)^{-2} = \bigo_{\Pr}(\lambda_K^2),
\end{align*}
and from Lemmas \ref{lemma:others} and \ref{lemma:elements_var}, we obtain
\begin{align*}
    (\delta_0 - \delta_\tau^{\top}\Delta_\tau^{-1}\delta - \delta^{\top}\Delta_\tau^{-1}\delta_\tau + \delta_\tau^{\top}\Delta_\tau^{-1} \Delta \Delta_\tau^{-1}\delta_\tau) = \bigo_{\Pr} \left(\frac{1}{n\pi_x(h)} \right) + \bigo_{\Pr}  \left(\frac{\lambda_K^{-1}}{n\pi_x(h)} \right),
\end{align*}
proving the result.
\end{proof}

\begin{proof}[Proof of Theorem 3.3]
    Using Lemmas \ref{lemma:S8_kNN}, \ref{lemma:S9_kNN} and \ref{lemma:elements_knn}, the proof is obtained similarly to the one of Theorems 3.1 and 3.2.
\end{proof}

\section{Quantile Regression with Functional Predictors}
\label{sec:appendix_qr}

For comparison against the proposed extremile regression model, in our simulation study we implemented a standard quantile regression with functional predictors. 

Consider observations $(Y_i, X_i(\cdot))_{i=1,..,n}$, and assume that the functional observations are zero mean without loss of generality. We estimate the eigenfunctions $\{ \phi_k(\cdot)\}_{k=1}^K$ of the functional setting via FPCA, so that we can write the finite representation od each $X_i(\cdot)$
$$X_i(s) \approx\sum_{k=1}^K c_{ik} \phi_k(s),$$
with $ c_{ik} = \int_{\mathcal{S}} X_i(s) \phi_k(s) \,\text{d}s$. The number of components, $K$, is chosen to explain a predetermined proportion of the total variance, in our simulation study is 95\%.

Then, for a given quantile level $\tau \in (0, 1)$, the conditional quantile of the response, $Q_{Y_i}(\tau | X_i)$ is
\begin{equation}
\label{eq:quat_reg}
    Q_{Y_i}(\tau | X_i) = \int_\mathcal{S} \beta(s) X_i(s) \,\text{d}s \approx \sum_{k=1}^{K} c_{ik} \int_\mathcal{S} \beta(s) \phi_k(s) \,\text{d}s.
\end{equation}
If the unknown functional coefficient is represented with the same basis functions, as in \eqref{eq:beta_approx}, the conditional quantile \eqref{eq:quat_reg} becomes
\begin{align*}
    Q_{Y_i}(\tau | X_i)  \approx \sum_{k=1}^{K} c_{ik} b_k,
\end{align*}
which is a linear quantile model with scalar covariates. 

The vector of coefficients, $(b_1(\tau), \dots, b_K(\tau))^T$, is estimated by minimizing the sum of asymmetrically weighted absolute residuals:
\begin{align*}
   \left(\widehat{b}_1(\tau), \dots, \widehat{b}_K(\tau)\right) = \underset{(b_1, \dots, b_K) \in \mathbb{R}^{K}}{\arg\min} \sum_{i=1}^{n} \rho_\tau \left( Y_i - \sum_{k=1}^{K} c_{ik} b_k \right).
\end{align*}
where $\rho_\tau(u)$ is the quantile regression check function:
\begin{align*}
    \rho_\tau(u) = u(\tau - \mathbb{I}(u < 0)) = 
    \begin{cases} 
        \tau u, & \text{if } u \ge 0, \\ 
        (\tau-1)u, & \text{if } u < 0 ,
    \end{cases}
\end{align*}
with $\mathbb{I}(\cdot)$ being the indicator function. This entire estimation procedure is repeated independently for each quantile level $\tau$ specified in the study.

The procedure was carried out in \texttt{R} using the function \texttt{fpca.face} in the package \texttt{refund} \citep{refund} for the first stage, and the \texttt{quantreg} package \citep{quantreg} for the second.

\end{document}